\documentclass[aps,prx,twocolumn,superscriptaddress,
groupedaddress]{revtex4}

\usepackage{dcolumn}
\usepackage{url}
\expandafter\let\csname equation*\endcsname\relax
\expandafter\let\csname endequation*\endcsname\relax
\usepackage{amsmath,amssymb}
\usepackage{amsthm}
\usepackage{graphicx,bm,color,hyperref}
\usepackage{mathrsfs}
\usepackage{dsfont}
\usepackage{mathtools}
\usepackage{bbm,bm}
\usepackage{mathdots}
\usepackage{xcolor}
\usepackage{tabularx}
\usepackage{booktabs}
\usepackage{hyperref}
\usepackage{amssymb}
\usepackage{amsfonts}
\usepackage{changes}
\usepackage{subfigure}
\usepackage{xcolor}
\usepackage{float}
\usepackage{cleveref}

\usepackage{tikz}

\usepackage{changes}

\newtheorem{definition}{Definition}
\newtheorem{theorem}{Theorem}

\newtheorem{lemma}[theorem]{Lemma}

\newtheorem{corollary}{Corollary}
\newtheorem{proposition}{Proposition}
\usepackage{cleveref}
\crefname{theorem}{Theorem}{Theorems}
\crefname{proposition}{Proposition}{Propositions}
\crefname{definition}{Definition}{Definitions}
\crefname{lemma}{Lemma}{Lemmas}
\crefname{figure}{Figure}{Figures}
\crefname{corollary}{Corollary}{Corollary}
\crefname{conjecture}{Conjecture}{Conjectures}
\crefname{section}{Section}{Sections}
\crefname{appendix}{Appendix}{Appendixes}
\crefname{observation}{Observation}{Observation}
\crefname{remark}{Remark}{Remark}
\crefname{example}{Example}{Examples}
\crefname{equation}{Eq.}{Eqs.}
\crefname{table}{Table}{Tables}
\crefname{theorem}{Theorem}{Theorems}

\newcommand\dloc{q} 
\DeclareMathOperator{\hiH}{\mathcal{H}} 

\DeclareMathOperator{\e}{\mathrm{e}}
\DeclareMathOperator{\iu}{\mathrm{i}}
\DeclareMathOperator{\Tr}{\mathrm{Tr}}

\newcommand{\ket}[1]{\vert #1 \rangle}
\newcommand{\bra}[1]{\langle #1 \vert}
\newcommand{\ketbra}[2]{\vert #1 \rangle\langle #2\vert}
\newcommand{\braket}[2]{\langle #1 \vert #2 \rangle}

\newcommand{\floor}[1]{\left\lfloor #1 \right\rfloor}
\newcommand\dH{d_H} 

\newcommand{\1}{\ensuremath{\mathbbm{1}}}
\renewcommand{\emph}{\textit}
\DeclareMathOperator{\rank}{rank}

\newcommand\uni {\mbox{-}\mathrm{UNI}}

\newcommand\Op{\mathcal{O}} 

\newcommand\nI{n^\mathrm{\star}} 
\newcommand\nII{n^\mathrm{\star \! \star}}
\newcommand\nIII{n^\mathrm{\star \! \star \! \star}} 
\newcommand\AI{A^\mathrm{\star}} 
\newcommand\AII{A^\mathrm{\star \! \star}} 
\newcommand\Aetc{A^\mathrm{\star \! \star \dots \star}} 
\newcommand\kI{k^\mathrm{\star}} 
\newcommand\kII{k^\mathrm{\star \! \star}} 

\begin{document}

\title{General stabilizer approach for constructing highly entangled graph states}
 \author{Zahra Raissi }
 \affiliation{Department of Physics, Virginia Tech, Blacksburg, VA 24061, USA}
\affiliation{Virginia Tech Center for Quantum Information Science and Engineering, Blacksburg, VA 24061, USA}
 
 \affiliation{ICFO-Institut de Ciencies Fotoniques, The Barcelona Institute of Science and Technology, Castelldefels (Barcelona), 08860, Spain}
 
 \author{Adam Burchardt}
\affiliation{QuSoft, CWI and University of Amsterdam, Science Park 123, 1098 XG Amsterdam, the Netherlands}
  \author{Edwin Barnes}
 \affiliation{Department of Physics, Virginia Tech, Blacksburg, VA 24061, USA}
\affiliation{Virginia Tech Center for Quantum Information Science and Engineering, Blacksburg, VA 24061, USA}

\begin{abstract}
Highly entangled multipartite states such as $k$-uniform ($k$-UNI) and absolutely maximally entangled (AME) states serve as critical resources in quantum networking and other quantum information applications. 
However, there does not yet exist a complete classification of such states, and much remains unknown about their entanglement structure.
Here, we substantially broaden the class of known $k$-UNI and AME states by introducing a method for explicitly constructing such states that combines classical error correcting codes and qudit graph states. This method in fact constitutes a general recipe for obtaining multipartitite entangled states from classical codes. Furthermore, we show that at least for a large subset of this new class of $k$-UNI states, the states are inequivalent under stochastic local operations and classical communication. This subset is defined by an iterative procedure for constructing a hierarchy of $k$-UNI graph states.
\end{abstract}

\maketitle

\section{Introduction}
Multipartite entanglement is at the very heart
of quantum information theory.
In recent years, significant effort has been devoted to characterizing the entanglement properties of multipartite quantum states and constructing new examples of highly entangled states \cite{three-qubit,four-qubit,three-qutrit,LOCC-symmetry,Scott2004,Zahra-min-support,Zahra-non-min-support,BurchardtRaissi20}. 
Much of this effort has focused on special classes of states such as graph states \cite{Hein,Hein-2006}, $k$-uniform ($k \uni$) and absolutely maximally entangled (AME) states \cite{Helwig-graphstates,Helwig-Cui,Zahra-min-support,Zahra-non-min-support}, which are critical resources for measurement-based quantum computing \cite{Raussendorf-Briegel}, quantum networking \cite{based-quantum-repeater,Helwig-Cui,photonic quantum repeaters}, and quantum error correction \cite{Hein-2006,Zahra-min-support,QECC-with-graph,Grassl-Graph,Stabilizer-Grph-codes,ZahraQECC}. 

Graph states are multipartite stabilizer states in which each vertex of a given graph represents a qudit, and the graph adjacency matrix defines the stabilizer generators \cite{Hein,Hein-2006}. $k \uni$ states are highly entangled pure states that have the property that all of their $k$-qudit reduced density matrices are maximally mixed \cite{Latorre,Preskill,ZahraQECC,Dardo-Karol,DardoQOA}. That is, the state $\ket \psi$ is a $k \uni$ state if 
 \begin{equation} \nonumber
   \rho_S = \Tr_{S^c} \ketbra\psi\psi \propto \1 \qquad \forall S  \subset \{1,\ldots,n\}, |S| \leq k \ ,
  \end{equation}
where $S^c$ denotes the complementary set of $S$. AME states correspond to the special case where $k=\floor{n/2}$, so that completely mixed states are obtained for any bipartition of the $n$ qudits. 

In light of the importance of graph and $k \uni$ states, it is interesting to consider what states lie at the intersection of these two classes. In addition to enabling interesting applications, $k$-UNI states that admit a graphical description can provide a powerful framework for investigating the entanglement structure of $k$-UNI states more generally. Moreover, the application of graph state techniques can facilitate finding new examples of $k$-UNI states, an endeavor that is generally not straightforward. So far, the only known systematic method of constructing $k \uni$ graph states is to start from a $k \uni$ state defined in terms of classical error correcting codes and to Fourier transform this into a graph state. This approach only yields a limited class of states corresponding to complete bipartite graphs \cite{Helwig-graphstates,Helwig-Cui,Zahra-min-support}. Whether or not a broader class of $k \uni$ graph states exists and how to find it has remained unclear (in this regards also see \cite{Zahra-non-min-support}).    

In this work, we uncover a large class of entangled states that are both graph states and $k \uni$ or AME states. We do this by finding a general set of constraints on the graph adjacency matrix that guarantee the resulting state is $k \uni$. These constraints allow us to go well beyond the special case of complete bipartite graphs. We show that explicit examples of such states can be obtained by constructing the adjacency matrix from classical error correcting codes. We then focus on a particular subset of these states that can be generated by applying an iterative procedure to produce a hierarchy of $k \uni$ and AME graph states. We show that at each level of this iterative process, the resulting state remains $k \uni$, but the states from different levels cannot be converted into each other using stochastic local operations and classical communication (SLOCC). 

The paper is organized as follows. 
In Sec.~\ref{sec:kuni-from-codes}, we first review the method of constructing $k \uni$ and AME states from classical linear codes, then in Sec.~\ref{sec:graph} we present our general method of finding $k \uni$ graph states.
In Sec.~\ref{sec:hierarchy}, we present the hierarchical graph state procedure.
In Sec.~\ref{sec:generalized-kuni}, we show that the generalized method of constructing $k \uni$ states from codes is the first level of this hierarchical procedure. 
Finally, in Sec.~\ref{sec:inequivalence} we show the states constructed at different levels of the hierarchy belong to different SLOCC classes.

\section{Constructing $k \uni$ states from classical codes}\label{sec:kuni-from-codes}
The connection between classical codes and $k \uni$ states has been shown to provide a systematic method of constructing a large set of $k\uni$ states \cite{Scott2004,Zahra-min-support,Zahra-non-min-support}. 
In this method, starting from a suitable classical code, a $k \uni$ state of $n$ qudits with local dimension $\dloc$ ($q$ level quantum systems also called qudits), is obtained by forming an equal superposition of the computational basis states corresponding to all of the codewords  $\ket{\vec{c}_i}$ of the code: 
  \begin{equation}\label{eq:state-n,0}
   \ket{\phi_{n,0}} = \sum_{i=1,\dots , \dloc^k} \ket{\vec{c}_i} 
   = \sum_{i=1,\dots , \dloc^k} \ket{c^{(i)}_1, \dots , c^{(i)}_n} \ ,
  \end{equation}
where we use that in the language of coding theory, a classical linear code encodes messages into a subset of codewords denoted by the vector $\vec{c}_i = (c^{(i)}_1, \dots , c^{(i)}_n)$ \cite[Chapter~1]{MacWilliams}  (see Appendix~\ref{app:classical linear codes} for the explicit definition and more details).
In the corresponding classical codes, we only consider $k \leq n/2$ as we focus on $k \uni$ states. The first index in $\ket{\phi_{n,0}}$ indicates the number of qudits $n$; the purpose of the second index (here equal to 0) will become clear later when we extend each of the states in Eq.~\eqref{eq:state-n,0} to a new family of $k \uni$ states. Here are some explicit examples of AME states constructed using the approach shown in Eq.~\eqref{eq:state-n,0}:
  \begin{equation}\label{eq:exaples-state-n,0}
  \begin{split}
    \ket{\phi_{2,0}} &=  \sum_{\alpha = 1,  \dots, \dloc} \, \ket{\alpha, \alpha} ,\ \ \ \ \ 
    \ket{\phi_{3,0}} = \sum_{\alpha = 1, \dots, \dloc} \, \ket{\alpha, \alpha ,\alpha}, \
  \end{split}
  \end{equation}
which are Bell and GHZ states, respectively.

How can we check if a given pure state is $k \uni$?
For this, we discuss two equivalent approaches.
The first is to check if all the reduced density matrices for up to $k$ qudits are maximally mixed. 
In this case, further information regarding the code parameters can be used to prove that the state in Eq.~\eqref{eq:state-n,0} is a $k \uni$ state (for more details see \cite{Zahra-min-support,Zahra-non-min-support}).
Another approach is based on the structure of the stabilizer formalism (see Appendix~\ref{app:stabilizer} for an overview of the stabilizer formalism): 
It is known that a pure stabilizer state of $n$ qudits is a $k \uni$ state if and only if, in all its stabilizer generators and arbitrary products of them, identity operators appear on at most $n-k-1$ different qudits (except for the trivial stabilizer operator given by the tensor product of $n$ identity operators).
In this method, one needs to check all possible products of the stabilizer generators. 
In the following, we will make heavy use of this second method to prove that a given graph state is $k \uni$.  
Although at first glance this method would seem to require checking exponentially many stabilizer operators, in practice it suffices to only check the stabilizer generators. This is because multiplying stabilizer operators does not increase the number of identity operators for the states we consider (more details can be found in Appendix~\ref{app:stabilizers-of-graph-codes}).

\section{Graph states}\label{sec:graph-states}\label{sec:graph}
Graph states are pure quantum states that are defined based on a graph. A graph $G=(V, \Gamma)$ is composed of a set $V$ of $n$ vertices (each qudit is represented by a vertex), and a set of weighted edges specified by the \emph{adjacency matrix} $\Gamma$ \cite{Nest,Hein,Hein-2006,Bahramgiri}.
$\Gamma$ is an $n \times n$ symmetric matrix such that $\Gamma_{i,j}=0$ if vertices $i$ and $j$ are not connected and $\Gamma_{i,j} >0$ otherwise. 
The graph state associated with a given graph $G$ is the $+1$ eigenstate of the following set of stabilizer generators  \cite{Nest,Hein,Hein-2006,Bahramgiri}: 
 \begin{equation}\nonumber
   S_i = X_i \prod_j (Z_j)^{\Gamma_{i,j}} , \qquad 1 \leq i \leq n \, ,
 \end{equation}
 where the operators $X$ and $Z$ are generalized Pauli operators acting on qudits with $\dloc$ levels.
 $X$ and $Z$ are unitary, traceless, and they satisfy the conditions $X^\dloc = Z^\dloc = \1$ and $ZX=\omega XZ$, where $\omega=\e^{\iu 2 \pi/\dloc}$ is a $\dloc$-th root of unity (see Appendix~\ref{app:stabilizer}).

We first briefly describe how we can convert the $k \uni$ state $\ket{\phi_{n,0}}$ in Eq.~\eqref{eq:state-n,0} into a graph state.
It has been shown in \cite{Zahra-non-min-support} that by performing local Fourier transforms $F_i = \sum_{j} \omega^{i j} \ketbra{j}{i} $ on all the last $n-k$ qudits of the $k \uni$ state $\ket {\phi_{n,0}} = \sum_{i} \ket{\vec{c}_i}$, the resulting state is a graph state corresponding to a \emph{complete bipartite graph}. An example of such a graph is depicted in part (a) of Table \ref{tabal-3-cases}. 
In this case, the adjacency matrix is 
 \begin{equation}\label{eq:adjacency-Gamma-n,0}
   \Gamma^{(n,0)} = \left[  \begin{array}{c|c}
    0 & -A^{\phantom{T}}  \\
    \hline
    -A^T & 0\\
 \end{array} \right] \ .
 \end{equation}
The matrix $A$ is directly related to the codewords in Eq.~\eqref{eq:state-n,0}. In particular, these codewords are obtained from 
  $\vec{c}_i = (\vec{x}_i \, G_{k \times n}) = (\vec{x}_i ,\, \vec{x}_i \,  A)$,
where $\vec{x}_i $ is a vector of size $k$, $G_{k\times n} = [\1_k |A]$, and $A$ is a matrix of size $k \times (n-k)$ that generates codewords. 
Note that, for a given classical code, several techniques are known for finding a suitable $A$ matrix \cite[Chapter~11]{MacWilliams} \cite{Zahra-min-support,Zahra-non-min-support} (also see Appendix~\ref{app:classical linear codes}). As an example, we can consider the $2 \uni$ state $\ket{\phi_{6,0}}$ (an explicit expression is given in Eq.~\eqref{eq:phi(6-nq)} below).
In this state the codewords are $\vec{c}_{i} = (\alpha , \beta , \alpha + \beta , \alpha + 2\beta  , \alpha + 3\beta , \alpha + 4\beta)$, where $\vec{x}_i = (\alpha , \beta)$ and 
 \begin{equation}\nonumber
   A= \left[  \begin{array}{cccc}
    1 & 1 & 1 & 1  \\
    1 & 2 & 3 & 4 \\
 \end{array} \right] \ .
 \end{equation}

Instead of starting from the $k \uni$ state in Eq.~\eqref{eq:state-n,0} and converting it into a graph state, we could alternatively start from a graph state with an adjacency matrix as in Eq.~\eqref{eq:adjacency-Gamma-n,0} and ask what properties must $A$ satisfy in order for this graph state to be $k \uni$? The answer is that every submatrix of $A$ must be nonsingular (see Appendix~\ref{app:prooftheorems} for details). 
It is generally challenging to find $A$ matrices that satisfy this property.
However, in the theory of classical error correcting codes, there is a systematic method that allows us to find them for some bounds (see Appendix~\ref{app:classical linear codes} for more details), and so we can use these matrices to create $k \uni$ graph states.
%

This alternative perspective in which we start from a graph and ask what properties must the adjacency matrix satisfy in order to obtain a $k \uni$ state allows us to uncover a larger class of $k \uni$ states. This larger class is summarized by the following theorem:
   \begin{theorem}\label{thm:general-adjacency}
    A graph state of $n$ qudits with local dimension $\dloc$ defined by the adjacency matrix 
    \begin{equation}\label{eq:general-adjacency}
       \Gamma^{n} = \left[  \begin{array}{c|c}
    0 & -A^{\phantom{T}}  \\
    \hline
    -A^T & \color{blue}B \\
 \end{array} \right] \ ,
    \end{equation}
    where $B$ is an arbitrary matrix and all submatrices of $A$ are nonsingular, is a $k \uni$ state.
  \end{theorem}
This theorem says that for every $k \uni$ graph state defined by the adjacency matrix $\Gamma^{(n,0)}$ from Eq.~\eqref{eq:adjacency-Gamma-n,0}, there is an infinite family of additional $k \uni$ states that can be obtained by replacing the lower-right block of zeros by an arbitrary matrix $B$ in the adjacency matrix. In Appendix~\ref{app:prooftheorems}, we prove  Theorem~\ref{thm:general-adjacency} by showing that any product of the stabilizer generators defined by $\Gamma^n$ has identity operators on at most $n-k-1$ qudits regardless of what $B$ is, so long as all the submatrices of $A$ are nonsingular.
Although Theorem~\ref{thm:general-adjacency} establishes sufficient conditions for a state to be $k$-UNI, we strongly suspect that the requirement that $A$ contains only nonsingular submatrices is also necessary for $k$-uniformity. We leave a rigorous proof of this to future work.

Given this new class of $k \uni$ states, a natural question to ask is whether all these states are locally equivalent or not. While this question is hard to answer for all possible choices of $B$, we can obtain a definite answer at least for a large set of $B$'s that exhibit a certain hierarchical structure. In what follows, we first describe this class of ``hierarchical graph states" and show how these states can be obtained from the $k \uni$ states of Eq.~\eqref{eq:state-n,0} by applying certain operators. We then demonstrate that these hierarchical graph states are inequivalent under SLOCC.

\section{Hierarchical graph states}\label{sec:hierarchy}
Now we are ready to define hierarchical graph states in which we consider particular forms of $B$ in the adjacency matrix $\Gamma^n$, Eq.~\eqref{eq:general-adjacency}.
Here, we will show that by following a recursive pattern, we can iteratively construct a series of different graph states that are all $k \uni$ states of $n$ qudits.
With each iteration, a new complete bipartite subgraph is included in the graph, and the number of edges increases. Starting from a given $k \uni$ state $\ket{\phi_{n,0}}$, we denote the new graph states obtained in each iteration as follows:
 \begin{equation}\label{eq:hierarchy_of_states}
    \ket{\phi_{n,\nI}} , \ \ket{\phi_{n,\nI, \nII}} ,  \ \ket{\phi_{n,\nI, \nII ,\nIII}} , \ \dots \,
 \end{equation}
where the total number of qudits $n$ remains fixed.
The number of $\star$'s indicates the level of the hierarchy, while $n^\star$, $n^{\star\star}$, ... are the numbers of qudits involved in the new bipartite subgraph that forms at that iteration.
In Table~\ref{tabal-3-cases} we present three levels of iteration. 
The qudits involved in the new bipartite subgraphs at the first and second iterations are shown in blue and red, respectively. 

The first iteration yields the $k \uni$ state $\ket{\phi_{n,\nI}}$.
This protocol proceeds as follows: We start from the adjacency matrix $\Gamma^{(n,0)}$ (Eq. \eqref{eq:adjacency-Gamma-n,0}) containing four blocks, two of which are matrices $-A$ and $-A^T$, and the remaining blocks are $0$'s.
In the first iteration, we replace part of one zero-block by the new adjacency submatrix $\Gamma^{(\nI,0)}$. More explicitly, we take the adjacency matrices:
\begin{equation}\nonumber
   \Gamma^{(n,0)} = \left[  \begin{array}{c|c}
    0 & -A^{\phantom{T}}  \\
    \hline
    -A^T & 0\\
 \end{array} \right], \ \text{and} \ 
    \Gamma^{(\nI,0)} = \left[  \begin{array}{c|c}
    0 & -{\AI}^{\phantom{T}}  \\
    \hline
    -{\AI}^T & 0\\
 \end{array} \right] \ ,
 \end{equation}
and construct the following new adjacency matrix
  \begin{equation}\label{eq:adjacency-Gamma-n,nq}
  \begin{split}
   \Gamma^{(n,\nI)} &=
    \left[  \begin{array}{c|c}
    0 & -A^{\phantom{T}}  \\
    \hline
    -A^T & 
        \color{blue}
         \begin{array}{c|c}    
     0 & 0\\
      \hline
     0 & \Gamma^{(\nI,0)}
      \end{array}
   \\
 \end{array} \right]  \\
    &=  
    \left[  \begin{array}{c|c}
    0 & -A^{\phantom{T}}  \\
    \hline
    -A^T & 
    \color{blue}
     \begin{array}{c|c}    
     0 & 0\\
      \hline
     0 & 
     \begin{array}{c|c}    
     0 & -\AI\\
      \hline
     -{\AI}^T & 0
      \end{array}   
      \end{array}
    \\
 \end{array} \right] \, .
   \end{split}
 \end{equation}
This yields the state $\ket{\phi_{n,\nI}}$, which is a $k \uni$ state of $n$ qudits.
Note that the size of the matrix $\Gamma^{(n,0)}$ is $n \times n$ while $\Gamma^{(\nI,0)}$ is a $\nI \times \nI$ matrix, and here we have assumed that $\nI \leq n-k$. 
In the case of $\nI=n-k$, the matrix $\Gamma^{(n,\nI)}$ can simply be written as
  \begin{equation}\nonumber
  \begin{split}
   \Gamma^{(n,\nI)} &=
    \left[  \begin{array}{c|c}
    0 & -A^{\phantom{T}}  \\
    \hline
    -A^T & 
        \color{blue}
          \Gamma^{(\nI,0)}
   \\
 \end{array} \right]  \\
    &=  
    \left[  \begin{array}{c|c}
    0 & -A^{\phantom{T}}  \\
    \hline
    -A^T & 
    \color{blue}
     \begin{array}{c|c}    
     0 & -{\AI}\\
      \hline
     -{\AI}^T & 0
      \end{array}   
    \\
 \end{array} \right] \, .
   \end{split}
 \end{equation}
In the corresponding graph, the number of edges increases, such that a new complete bipartite graph with $\nI$ vertices forms inside the original graph state, as shown in part (b) of Table~\ref{tabal-3-cases}.
This procedure is described in further detail using the stabilizer formalism in Appendix~\ref{app:stabilizers-of-graph-general}. 

This procedure continues such that at every iteration, we replace some of the zeros of the last diagonal block of the previous step with a new adjacency submatrix.
To continue our example from above, in the second iteration, in order to construct the $k \uni$ state $\ket{\phi_{n,\nI,\nII}}$, we take the adjacency matrix $\Gamma^{(\nII,0)} = \left[  \begin{array}{c|c}
    0 & -{\AII}  \\
    \hline
    -{\AII}^T & 0\\
 \end{array} \right] $
that corresponds to the $\kII \uni$ graph state $\ket{\phi_{\nII,0}}$ and insert this into the adjacency matrix from the first iteration:
  \begin{equation}
  \label{eq:adjacency-Gamma-n,nI,nII}
  \begin{split}
   &\Gamma^{(n,\nI,\nII)} =  
    \left[  \begin{array}{c|c}
    0 & -A^{\phantom{T}}  \\
    \hline
    -A^T & 
    \color{blue}
     \begin{array}{c|c}    
     0 & 0\\
      \hline
     0 & 
     \begin{array}{c|c}    
     0 & -{\AI}\\
      \hline
     -{\AI}^T & 
            \color{red}
         \begin{array}{c|c}    
     0 & 0\\
      \hline
     0 & \Gamma^{(\nII,0)}
      \end{array}
      \end{array}   
      \end{array}
    \\
 \end{array} \right] 
 \\
 &= \left[  \begin{array}{c|c}
    0 & -A^{\phantom{T}}  \\
    \hline
    -A^T & 
    \color{blue}
     \begin{array}{c|c}    
     0 & 0\\
      \hline
     0 & 
     \begin{array}{c|c}    
     0 & -{\AI}\\
      \hline
     -{\AI}^T & 
            \color{red}
         \begin{array}{c|c}    
     0 & 0\\
      \hline
     0 & 
     \begin{array}{c|c}    
     0 & -{\AII}\\
      \hline
     -{\AII}^T & 0
      \end{array}   
      \end{array}
      \end{array}   
      \end{array}
    \\
 \end{array} \right] \ .
   \end{split}
 \end{equation}
The above adjacency matrix represents the state $\ket{\phi_{n,\nI , \nII }}$, and the associated graph is shown in part (c) of Table \ref{tabal-3-cases}. 
This procedure also applies equally well to higher orders; here we focus on the 0th and 1st orders of the hierarchy only for the sake of simplicity.
Therefore, iterating this process further, we can in the next step construct the state $\ket{\phi_{n,\nI, \nII ,\nIII}}$.
And in general, by continuing to insert smaller and smaller submatrices into the adjacency matrix:
    \begin{equation}\label{eq:adjacency-Gamma-n,nI,nII,etc}
   \begin{split}
  &\Gamma^{(n,\nI , \nII , \dots)} =
  \\
  & \left[  \begin{array}{c|c}
    0 & -A^{\phantom{T}}  \\
    \hline
    -A^T & 
    \color{blue}
     \begin{array}{c|c}    
     0 & 0\\
      \hline
     0 & 
     \begin{array}{c|c}    
     0 & -{\AI}\\
      \hline
     -{\AI}^T & 
            \color{red}
   \begin{array}{c|c}    
     0 & 0\\
      \hline
     0 & 
         \begin{array}{cc}    
  \ddots & \\
   & 
   \color{violet}
     \begin{array}{c|c}    
     0 & -{\Aetc}\\
      \hline
     -{\Aetc}^T & 0
      \end{array}   
      \end{array}
      \end{array}
      \end{array}   
      \end{array}
    \\
 \end{array} \right] ,
   \end{split}
  \end{equation}
each time we add more edges while preserving the number of vertices/qudits.

\section{generalizing the method of constructing states from codes}\label{sec:generalized-kuni}
Next, we show that the state obtained after the first iteration of the hierarchical procedure described in the previous section, $\ket{\phi_{n,\nI}}$, can also be obtained directly from Eq.~\eqref{eq:state-n,0} by applying certain operators to it. This observation allows us to show that some of the states in the hierarchy are inequivalent under SLOCC. We can additionally view the result of this section as a more general method for constructing $k \uni$ and AME states directly from classical codes compared to the method summarized by Eq.~\eqref{eq:state-n,0}. Here, the codewords act as a more general resource, rather than simply entering into an equal superposition as in Eq.~\eqref{eq:state-n,0}.
In the following we discuss this in more detail and provide closed-form expressions for these new $k \uni$ states.

To find the closed-form expression of the state $\ket{\phi_{n,\nI}}$, we first introduce operator $\Op_{\nI}$ based on its action on a given product state $\ket{i_1,\dots , i_{\nI}} $: 
  \begin{equation}\label{eq:Operator}
   \begin{split}
   & \Op_{\nI} \ket{i_1,\dots , i_{\nI}} \coloneqq \\
   & Z^{-i_1} \otimes \dots \otimes Z^{-i_{\kI}} \otimes X^{i_{\kI+1}} \otimes \dots \otimes X^{i_{\nI}} \ket{\phi_{\nI,0}} \ ,
     \end{split}
  \end{equation}
where $\ket{\phi_{\nI,0}}$ is a $\kI \uni$ state of $\nI$ qudits as in Eq.~\eqref{eq:state-n,0}.
Note that, in this operation, the number of $Z$ operators is equal to $\kI$, while the number of $X$ operators is $\nI-\kI$.

Now we use operator $\Op_{\nI}$ to present the general method of constructing $k \uni$ states $\ket{\phi_{n,\nI \geq 2}}$ from classical codes.

  \begin{proposition}\label{thm:constructing-state-n,nI}
    Consider a $k \uni$ state $\ket{\phi_{n,0}}$ constructed from a classical linear code according to Eq.~\eqref{eq:state-n,0}, and the operator $\Op_{\nI}$ based on the definition given in Eq.~\eqref{eq:Operator}.
    The pure state
     \begin{equation} \label{eq:state-n-nI}
      \ket{\phi_{n,\nI}} \coloneqq \1^{(n-\nI)} \otimes \Op_{\nI} \ \ket{\phi_{n,0}} \, ,
     \end{equation}
   is a $k \uni$ state for any $2 \leq \nI \leq n-k$. 
  \end{proposition}
In this context, we can now view the indices of the state $\ket{\phi_{n,\nI}}$ as follows: The first index indicates the number of qudits, while the second index indicates the number of non-trivial local operators acting on the state.
In the case where $\nI=0$, the state is given by a superposition of all the codewords with phases all equal to $+1$ as in Eq.~\eqref{eq:state-n,0}. As an example, we present an explicit formula for constructing $2 \uni$ states $\ket{\phi_{6,2}}$, and $\ket{\phi_{6,3}}$ from a $2 \uni$ state $\ket{\phi_{6,0}}$ constructed from classical codes:
\begin{widetext}
  \begin{align} \label{eq:phi(6-nq)}
   \begin{split}
     \ket{\phi_{6,0}} &= \sum_{\alpha, \beta}
   \ket{\alpha , \beta , \alpha + \beta , \alpha + 2\beta  , \alpha + 3\beta , \alpha + 4\beta }, \\
   \ket{\phi_{6,2}} &= \1^{(4)} \otimes \Op_2 \ \ket{\phi_{6,0}}  = 
    \sum_{\alpha, \beta}
   \ket{\alpha , \beta , \alpha + \beta , \alpha + 2\beta} \otimes Z^{-(\alpha + 3\beta)} \otimes X^{(\alpha+4\beta)} \, \ket{\phi_{2,0}},\\
   \ket{\phi_{6,3}} &= \1^{(3)} \otimes \Op_3 \ \ket{\phi_{6,0}}  = 
    \sum_{\alpha, \beta}
   \ket{\alpha , \beta , \alpha + \beta } \otimes Z^{-(\alpha + 2\beta)} \otimes X^{(\alpha+3\beta)} \otimes X^{(\alpha+4\beta)} \, \ket{\phi_{3,0}},\\
  %
 \, ,
  \end{split}
\end{align}  
\end{widetext}
where the local dimension is $\dloc =5$, and the states $\ket{\phi_{2,0}}$, and $\ket{\phi_{3,0}}$ are given in Eq.~\eqref{eq:exaples-state-n,0}. 

In Appendix~\ref{app:stabilizers-of-graph-general} we provide a full stabilizer description of the state $\ket{\phi_{n,\nI}}$ and then show this state is local-unitary-equivalent to the graph state obtained from $\Gamma^{(n,\nI)}$.

\section{Inequivalence under SLOCC}\label{sec:inequivalence}
Now we show that the states obtained at the first two levels of the iteratively constructed hierarchy belong to different SLOCC classes.
We first consider $k \uni$ states with $k < n/2$ and show that two states constructed at the first level of the hierarchy belong to different SLOCC classes, as summarized by the following proposition:

  \begin{proposition}\label{prop:LOCC-firsthierarchy} 
Two $k \uni$ states $\ket{\phi_{n,0}}$ and $\ket{\phi_{n,{\nI}}}$ defined by Eq.~\eqref{eq:state-n-nI} belong to different SLOCC classes.
  \end{proposition}

We first note that the Schmidt-rank vector of a multipartite pure state cannot be modified by SLOCC protocols \cite{rank}.
Also we know that for the state $\ket{\phi_{n,0}}$ and any subset $ S \subset \{1,\dots ,n\}$, the rank of the reduced density matrix satisfies \cite{Zahra-min-support,Zahra-non-min-support}
 \begin{equation} 
  \rank (\rho_S) \leq \dloc ^k \ ,
 \end{equation}
where $\rho_S=\Tr_{S^c} \ketbra{\phi_{n,0}}{\phi_{n,0}}$.
In order to prove proposition~\ref{prop:LOCC-firsthierarchy}, we show in Appendix \ref{app:inequivalence-SLOCC} that for specific subsets of size $|S|= k+ \kI \leq n/2$, the reductions are maximally mixed, 
where $k$ lies in the support of the first $n- \nI$ qudits, and $\kI$ lies in the last $\nI$ qudits.

Next we discuss the AME states that are special among multipartite entangled states and particularly interesting to study in terms of local equivalence classes.  
The long-standing question of whether or not it is possible for two AME states to belong to two distinct SLOCC classes was settled in Ref.~\cite{BurchardtRaissi20}, although no explicit method of construction were provided.
Here, we show this explicitly for two large sets of AME states constructed using the iterative procedure we introduced:
One is the state $\ket{\phi_{n,0}}$ from Eq.~\eqref{eq:state-n,0} that is constructed from classical codes, while the second state is $\ket{\phi_{n,\nI=2}}$, which is obtained at the first level of the hierarchy.
In both cases, we assume that the number of qudits $n$ is odd, in which case the following holds:

\begin{proposition} 
\label{TheoremAppF}
For any odd number of qudits $n$, the following two AME states, $\ket{\phi_{n,0}}$ and $\ket{\phi_{n,2}}$, are not SLOCC equivalent.
\end{proposition}
We prove this proposition in Appendix \ref{app:inequivalence-SLOCC-AME}. 
We have not been able to prove the above proposition for all AME states constructed at different levels of the hierarchy (see also Corollary 1. in \cite{BurchardtRaissi20}).
We leave this question open for now. 

\section{Conclusion}
In this work, we presented a general method for constructing highly entangled $k \uni$ and AME states using the stabilizer formalism in conjunction with classical error correcting codes. This method significantly expands the set of known $k \uni$ and AME states. We further showed that a special subset of this new class of entangled states can be obtained from an iterative procedure that produces a hierarchy of $k \uni$ graph states. For the first two levels of the hierarchy, we showed that while the states share the same number of qudits and the same $k \uni$ property, they cannot be converted into each other by means of SLOCC operations. 
 We also presented new
sets of AME states that belong to different SLOCC and, equivalently, LU classes. These results constitute a general approach to constructing new examples of $k \uni$ and AME states that admit graphical representations, providing further instances of these important resource states and an opportunity to shed more light on the entanglement structure of these highly entangled multipartite states.
Although we have only proven that certain conditions on the underlying classical codes are sufficient for $k$-uniformity, we anticipate that future work will demonstrate rigorously that these conditions are also necessary.

\acknowledgements

We would thank, Antonio Ac\'{i}n, Jens Eisert, Mario Flory, Markus Grassl, Barbara Kraus, and Karol \. Zyczkowski for discussions and useful comments.
This research is supported by the National Science Foundation (grant nos. 1741656 and 2137953), FIS2020-TRANQI and Severo Ochoa CEX2019-000910-S), Fundaci\'o Cellex, Fundaci\'o Mir-Puig, Generalitat de Catalunya (CERCA Program) and AdG CERQUTE,
and the National Science Center in Poland under the Maestro grant number DEC-2015/18/A/ST2/00274, AB acknowledges the support of an NWO Vidi grant (Project No. VI.Vidi.192.109). 

\begin{widetext}

\begin{table}[t]
\begin{center}
\begin{tabular}{ |c|c|c|c|}  
 \hline
 &$k \uni$ state & Adjacency matrix & Graph state   \\
 \hline \hline
     (a) & $\ket{\phi_{n,0}}$ &  
   $\Gamma^{(n,0)}$ = $\left[  \begin{array}{c|c}
    0 & -A  \\
    \hline
    -A^T & 0\\
 \end{array} \right] $
      &  
      \parbox[c]{20em}{\includegraphics[width=3.4in]{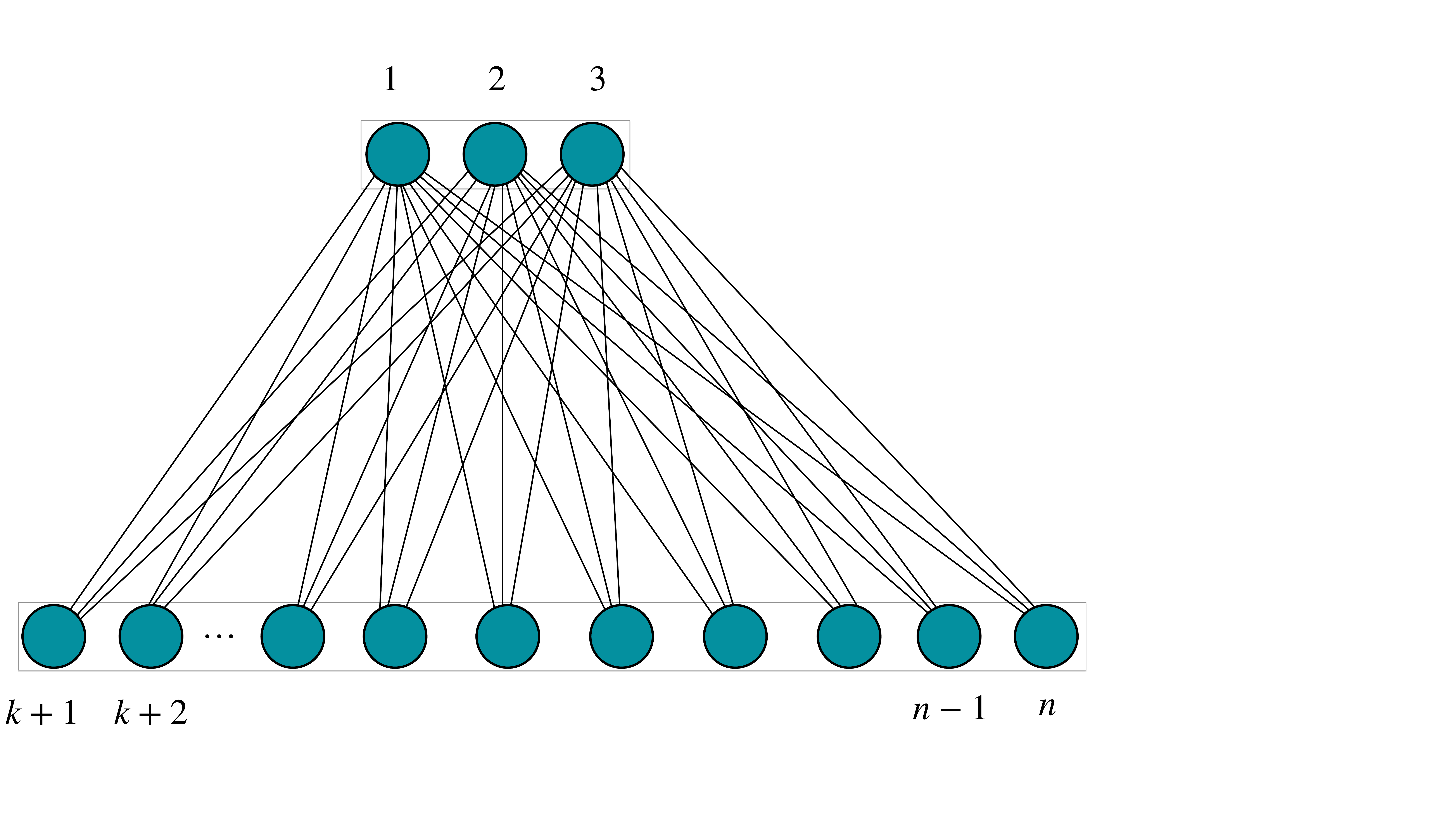}}
         \\ 
 \hline
     (b) &$\ket{\phi_{n,\nI}}$ & 
        $\Gamma^{(n,\nI)}$ =  
    $\left[  \begin{array}{c|c}
    0 & -A^{\phantom{T}}  \\
    \hline
    -A^T & 
    \color{blue}
     \begin{array}{c|c}    
     0 & 0\\
      \hline
     0 & 
     \begin{array}{c|c}    
     0 & -{\AI}\\
      \hline
     -{\AI}^T & 0
      \end{array}   
      \end{array}
    \\
 \end{array} \right]$
       & 
 \parbox[c]{20em}{\includegraphics[width=3.9in]{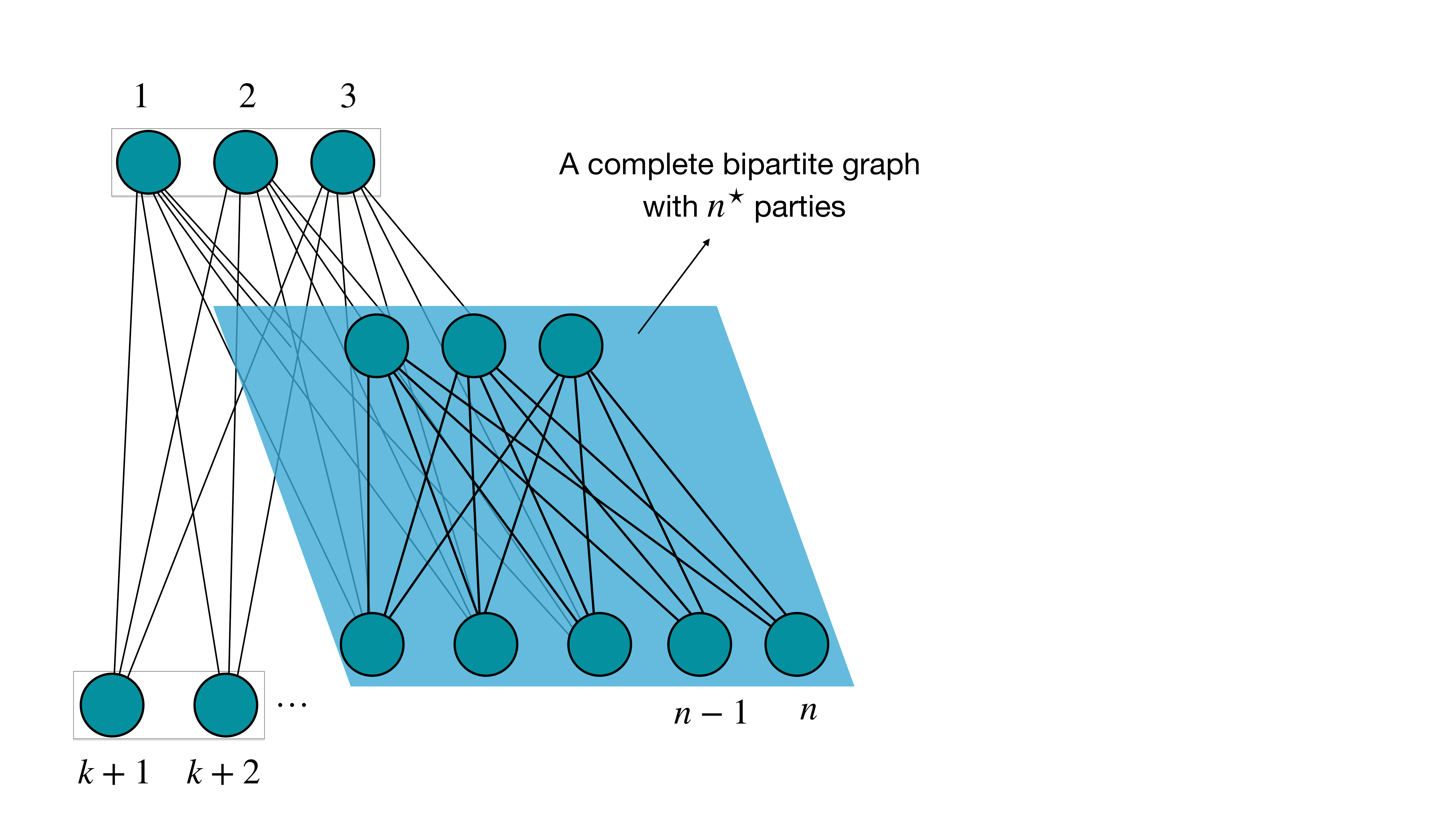}}
           \\ 
  \hline
    (c) & $\ket{\phi_{n,\nI,\nII}}$ & 
   $\Gamma^{(n,\nI,\nII)}$ 
   =  
    $\left[  \begin{array}{c|c}
    0 & -A^{\phantom{T}}  \\
    \hline
    -A^T & 
    \color{blue}
     \begin{array}{c|c}    
     0 & 0\\
      \hline
     0 & 
     \begin{array}{c|c}    
     0 & -{\AI}\\
      \hline
     -{\AI}^T & 
            \color{red}
         \begin{array}{c|c}    
     0 & 0\\
      \hline
     0 & 
     \begin{array}{c|c}    
     0 & -{\AII}\\
      \hline
     -{\AII}^T & 0
      \end{array}   
      \end{array}
      \end{array}   
      \end{array}
    \\
 \end{array} \right] $
     &   
    \parbox[c]{22em}{\includegraphics[width=4.0in]{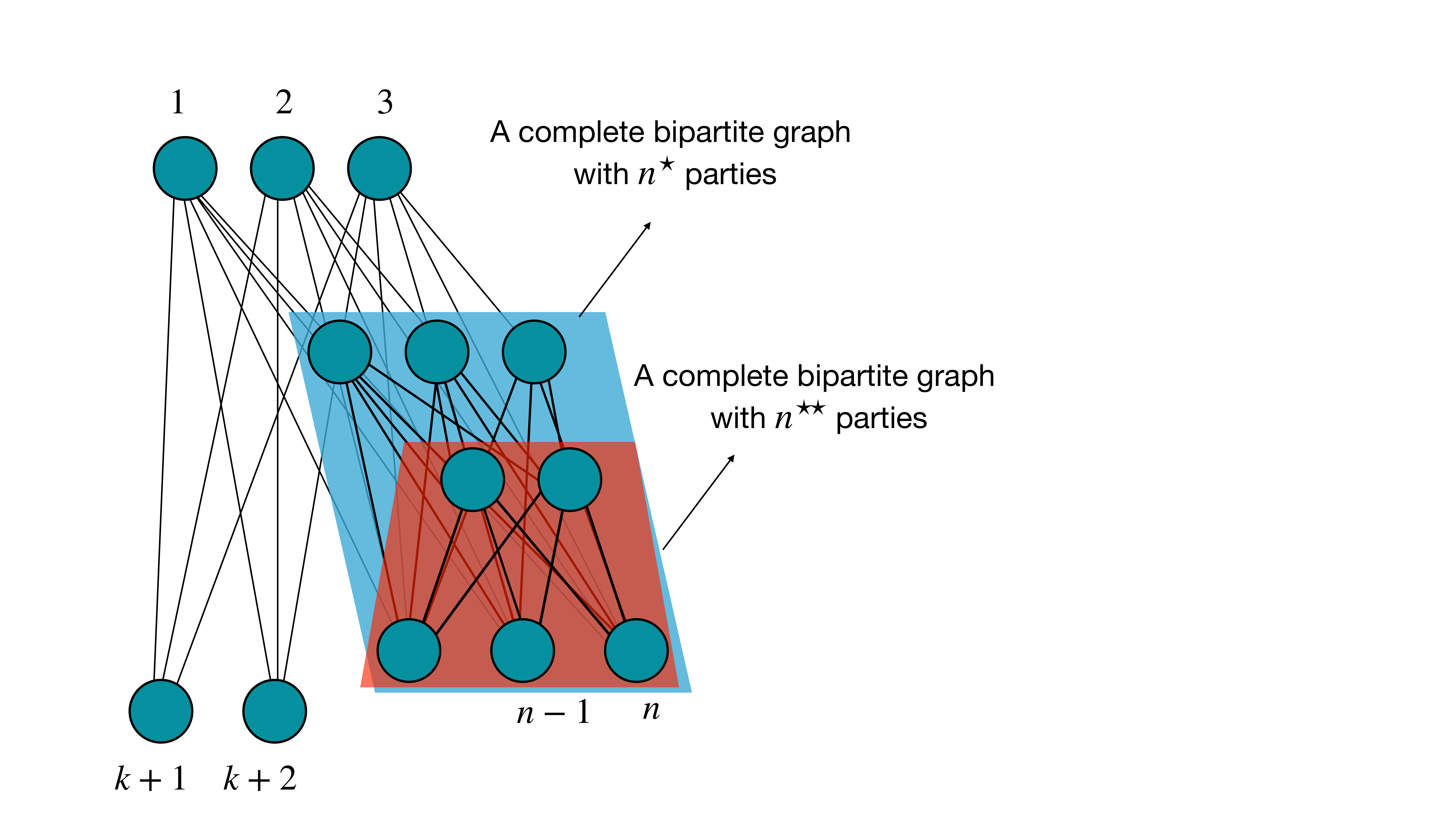}} 
      \\ 
 \hline
\end{tabular}
\end{center}
 \caption{\label{tabal-3-cases} 
 Different levels of the hierarchical graph construction that yields $k \uni$ graph states. (a) State, adjacency matrix, and graph obtained directly from a classical code. (b) The $k \uni$ graph state obtained at the first level of the hierarchy. (c) The $k \uni$ graph state obtained at the second level of the hierarchy. In the first and second levels of the hierarchy, the blue part of the graph contains $\nI$ qudits, and the red part contains $\nII$ qudits.}
\end{table}

\end{widetext}

\appendix
\begin{widetext}

\section{Explicit definition of classical linear codes}\label{app:classical linear codes}
In general, a classical linear error-correcting code is denoted by $\left[n ,k,\dH \right]_\dloc$, when it encodes $\dloc ^k$ messages into codewords living in a larger space of dimension $\dloc ^n$, all having Hamming distance at least $\dH$.
Linear codes are a special class of codes whose injective map from the set of messages to the set of codewords is linear and defined over a finite field $GF(\dloc)$ (for the motivation for using finite fields see \cite[Chapter~3]{MacWilliams}).
Codewords of a linear code are constructed by taking linear combinations of the rows of a matrix called the generator matrix $G_{k\times n}$.
For a given vector $\vec{x}_i$, a codeword can be written as $\vec{c}_i = \vec{x}_i \, G_{k\times n}$.
A generator matrix can always be written in the standard form  
 \begin{equation}\label{eq:generator}
   G_{k \times n} = [\1_k |A]\, , 
 \end{equation}
where $\1_k$ is a $k \times k$ identity matrix, and $A$ is a $k \times (n-k)$ matrix with elements in $GF(\dloc)$. 
Every linear code $C=[n,k,\dH]_\dloc$ has a dual code $C^\perp$ defined such that its codewords are orthogonal to all the codewords of the original code with respect to the standard Euclidean inner product of the finite field \cite[Chapter~1]{MacWilliams}.
The generator matrix of the dual code is the so-called parity check matrix $H$.
It satisfies $GH^T = 0$. 

A $k \uni$ state $\ket {\phi_{n,0}}$, Eq.~\eqref{eq:state-n,0}, can be constructed by taking a superposition of the computational basis states corresponding to all of the codewords of a linear code $C$ such that it and its dual $C^\perp$ have a minimum distance of at least $k+1$.
Using Eq.~\eqref{eq:state-n,0}, we have
 \begin{equation}\label{eq:kunimin-general}
   \ket{\phi_{n,0}}=
   \sum_i \ket{\vec{c}_i}=
   \sum_{i} \ket{\vec{x}_i \, G_{k  \times n}}
   =\sum_{i} \ket{\vec{x}_i,\, \vec{x}_i A}\, ,
 \end{equation}
Note that the linear code $C$ with a given matrix $A$ is an MDS (maximum distance separable) code if and only if every square submatrix of $A$ is nonsingular \cite[Chapter~11]{MacWilliams}, \cite{Singleton}.
The MDS codes are those linear codes that achieve maximum possible minimum Hamming distance \cite{Singleton} \cite[Chapter~11]{MacWilliams}:
 \begin{equation}\label{eqapp:Singleton}
  \dH \leq n -k+1\ .
 \end{equation}

Now let us discuss how to construct suitable $A$ matrices.
For this, we first need to introduce the concept of \emph{Singleton arrays} \cite{Singleton,Roth}\cite[chapter~11]{MacWilliams}.
Any finite field $GF(q)$, with $q$ a power of a prime number, contains at least one primitive element \cite[chapter~4]{MacWilliams}.
An element $\gamma \in GF(q)$ is called primitive if all the nonzero elements of $GF(q)$ can be written as some integer power of $\gamma$.
Given any such primitive element $\gamma$, the Singleton array of size $q$ is defined as
\begin{equation}
S_q \coloneqq \begin{array}{ccccccc}
1 & 1 & 1 & \ldots & 1 & 1 & 1\\
1 & a_1 & a_2 & \ldots & a_{q-3} & a_{q-2} &  \\
1 & a_2 & a_3  & \ldots & a_{q-2} &   &  \\
\vdots & \vdots & \vdots & \iddots &  &   &  \\
1 & a_{q-3} & a_{q-2} &   &   &   &    \\
1 & a_{q-2} &  &   &   &   &   \\
1 &   &  &   &   &   &    \\ 
 \end{array}, \label{sq}
\end{equation}
with
\begin{equation}
  a_i \coloneqq \frac{1}{1-\gamma^i} .
\end{equation}
It follows that by taking rectangular submatrices of $S_q$, it is hence possible to construct proper $A$ matrices \cite[chapter~11]{MacWilliams}.
All one has to do is to take a power of a prime $q$ sufficiently large such that $S_q$ contains a submatrix of the required size, and then take this as the matrix $A$ in Theorem 1.

As an example, let us consider the case $q=5$.
Taking $\gamma=3$, which is a primitive element in $GF(5)=\{0,1,2,3,4\} \mod(5)$, we find
\begin{align}
a_1 &= \frac{1}{1-3}= \frac{1}{3} = 2 \\
a_2 &= \frac{1}{1-9}=\frac{1}{2}=3 \\
a_3 &= \frac{1}{1-27}=\frac{1}{4}=4 \ ,
\end{align}
and obtain
\begin{equation}
S_5 = \begin{array}{ccccc}1 & 1 & 1 & 1 &1 \\ 1 & 2 & 3 & 4 & \\ 1 & 3 & 4 &   & \\ 1 & 4 &  &  & \\ 1 &   &   &   &  \end{array}.
\end{equation}
The biggest submatrix has size $3\times 3$.
Hence, taking
\begin{equation}
A_{3\times 3}=\left[\begin{array}{ccc}
1 & 1 & 1 \\
1 & 2 & 3 \\
1 & 3 & 4  \\
 \end{array} \right] \ ,
\end{equation}
we can construct a matrix $A$, and the resulting AME state can precisely be constructed using the method discussed in Theorem 1.
In the Appendix of Ref.~\cite{Zahra-min-support}, more details on the explicit construction of Singleton arrays are available. We also note that the main conjuncture for the MDS codes emphasise that this is the method of constructing $A$ matrices  for the following interval \cite{Hirschfeld-ConjectureMDS,Singleton,Roth}: 
\begin{equation}
n=
  \begin{cases}
    \dloc+2 &\text{for}\ k=3 \text{and} \ k=\dloc-1  \ \text{both with $\dloc$ even}\\
    \dloc+1  &\text{in all other cases}
  \end{cases} .
\end{equation}

\section{Stabilizer formalism}\label{app:stabilizer}
The stabilizer formalism is a useful tool in different branches of quantum information science like quantum error correcting codes \cite{Gottesman-thesis,Zahra-min-support,ZahraQECC}, one-time or cluster states \cite{Raussendorf-Briegel}, and graph states \cite{Hein}. 
We first recall the definition of the generalized Pauli operators acting on a $\dloc$-dimensional Hilbert space:
\begin{align} \nonumber
  X\ket{j}&=\ket{j+1 \mod \dloc}, \\
  Z\ket{j}&=\omega^j\,\ket{j}\, , \nonumber
\end{align}
where $\omega \coloneqq \e^{\iu 2 \pi/\dloc}$ is the $\dloc$-th root of unity. 
$X$ and $Z$ are unitary, traceless, and they satisfy the conditions $X^\dloc = Z^\dloc = \1$ and $ZX=\omega XZ$.
For a collection of $n$ qudits, we shall use
subscripts to identify the corresponding Pauli operators. 
For example, $Z_i$ and $X_i$ operate on the space of $i$-th qudit. 
Operators of the form
\begin{equation}
\label{pauliProd}
\omega^\lambda X_1^{w_1} Z_1^{w'_1} 
\otimes \cdots \otimes
X_n^{w_n} Z_n^{w'_n} \ ,
\end{equation}
are called \textit{Pauli products}, where $\lambda, w_i$, and $w'_i \in \mathbb{Z}_\dloc$ for all $i \in \{1, \dots n \}$. For a given number of qudits $n$, the collection
of all possible Pauli products (\ref{pauliProd}) form a group called the Pauli group $P_n$. 

Now we are ready to define the stabilizer formalism.
For a given state $\ket \psi $ of $n$ qudits, the element $S \in P_n$ is called a \textit{stabilizer} operator of a state if it leaves the state invariant, i.e. $S \ket \psi =\ket \psi$. 
The set of all stabilizer operators of a state $\ket \psi $ is denoted by $\mathcal{S}_{\ket \psi}$.  
If $| \mathcal{S}_{\ket \psi} |=\dloc^n$, the state $\ket \psi $ is called a \textit{stabilizer state}, and its density matrix has the following representation: \cite{Stabilizers} 
\begin{equation}
\label{formOfStab}
\ketbra {\psi}{\psi} = \dfrac{1}{\dloc^n}
\sum_{S \in \mathcal{S}_{\ket \psi}} S \ .
\end{equation}
Each stabilizer state has exactly $n$ independent and commuting stabilizers, called \emph{stabilizer generators}, $S_i \in \mathcal{S}_{\ket \psi}$, such that any operator $S \in \mathcal{S}_{\ket \psi}$ is of the form
\begin{equation}
S = S_1^{w_1}\cdots S_n^{w_n}\, ,
\end{equation}
for some choice of $w_i \in \mathbb{Z}_\dloc$, $i \in \{1,\dots n\}$.

For a stabilizer state, there is a straightforward method to calculate the form of a reduced density matrix $\rho_C$. 
Consider any subsystem $C \subset \{1, \dots, n\}$ of $n$ qudits and an element of the Pauli group $S\in P_n$. By $S_{|C} $ we denote the restriction of the operator $S$ to the subsystem $C$. 
For example, the operator $S=X_1\otimes X_2Z_2^2$ restricted to the subsystem $\{2\}\subset\{1,2\}$ has the form $S_{|\{2\}}= X_2Z_2^2$. Using Eq.~\eqref{formOfStab} and the fact that Pauli operators $X_i^{x_i} Z_i^{z_i}$ are traceless except when $x_i=z_i=0$, one can show the following:

\begin{proposition}\label{propStab} 
Consider a stabilizer state $\ket \psi$ of $n$ qudits with local dimension $\dloc$ and the stabilizer operators $\mathcal{S_{\ket \psi}} \subset P_n$. Consider any subsystem $C \subset \{1,\dots , n\}$ of the $n$ qudits and its complementary subsystem $C^c= [n] / C$, where $[n]=\{1,\dots , n\}$. 
The reduced density matrix $\rho_C$ has the following form:
\[
\rho_C = \Tr_{C^c} \ketbra{\psi}{\psi} \propto \, \Tr_{C^c} \sum_{\substack{ S\in \mathcal{S}_{\ket{\psi}} \\ S_{|C^c} = \1^{\otimes |C^c|}}} S \ .
\]
In particular, $\rho_C$ is maximally mixed, i.e. $\rho_C\propto \1$, if and only if for any $S \in \mathcal{S}_{\ket{\psi}}$,  $S_{|C^c} = \1^{\otimes |C^c|}$ implies $S= \1^{\otimes n}$.
\end{proposition} 

\noindent

Observe that the following corollary is an immediate consequence of the above proposition, because the only stabilizer operators $S\in \mathcal{S}_{\ket{\psi}}$ that contribute after the partial trace are those that have identity operators acting on all qudits belonging to the complementary set $C^c$ (see also \cite{EntMIxedstates}).

  \begin{corollary}\label{prop:number-of-1}
A pure stabilizer state of $n$ qudits is a $k \uni$ state if and only if in its stabilizer operators the identity matrix appears at most on $n-k-1$ different qudits.
We note that there is one stabilizer operator that is formed by tensoring $n$ identity matrices.
  \end{corollary}

\section{Stabilizer formalism of the $k \uni$ states}\label{app:stabilizers-of-graph} 
\subsection{The states obtained from classical codes} \label{app:stabilizers-of-graph-codes}
Now we discuss the stabilizer formalism of the $k \uni$ states $\ket{\phi_{n,0}}$  given in Eq.~\eqref{eq:state-n,0}.
Remember that, given a classical linear code with codewords $\vec{c}_i$ for $i \in \{1, \dots , n\}$, it is possible to construct a $k \uni$ state $\ket{\phi_{n,0}}$. 
The codewords can be obtained from $\vec{c}_i = \vec{x}_i \, G_{k \times n}$ where $G_{k \times n} = [\1_k |A]$ (see Appendix~\ref{app:classical linear codes}).
Denoting the matrix elements of $A$ by $a_{i,j}$, it is straightforward to see that for $q$ prime, the stabilizer generators of state $\ket{\phi_{n,0}}$ have the following form (see also \cite{Zahra-min-support} and \cite{Zahra-non-min-support}):
 \begin{equation} 
  S^{\ket{\phi_{n,0}}}  \coloneqq  \\
   \begin{cases}
    \bigotimes_{j=1}^{i-1} \1 \otimes X  \bigotimes_{j=i+1}^{k}  \1 \bigotimes_{j=1}^{n-k} X^{a_{i,j}} & \qquad 1\leq i \leq k \\
    \phantom{a} \\
    \bigotimes_{j=1}^{k} Z^{-a_{j,i}} \bigotimes_{j=1}^{i-1} \1 \otimes Z  \bigotimes_{j=i+1}^{n}  \1 & \qquad 1 < i \leq n-k
  \end{cases} .
 \end{equation} 
If one performs a local Fourier transform on all the last $n-k$ qudits, mapping $Z$ operators into $X$ operators and vice versa, the stabilizer generators become
 \begin{equation}  
  S^{\ket{\phi_{n,0}}}  \coloneqq  \\
   \begin{cases}
    \bigotimes_{j=1}^{i-1} \1 \otimes X  \bigotimes_{j=i+1}^{k}  \1 \bigotimes_{j=1}^{n-k} Z^{-a_{i,j}} & \qquad 1 \leq i \leq k \\
    \phantom{a} \\
    \bigotimes_{j=1}^{k} Z^{-a_{j,i}} \bigotimes_{j=1}^{i-1} \1 \otimes X \bigotimes_{j=i+1}^{n}  \1 & \qquad 1 < i \leq n-k
  \end{cases} .
 \end{equation}
In this case, the resulting state is a graph state corresponding to a \emph{complete bipartite graph} with the following set of stabilizer generators:
  \begin{equation}
  \label{egT}
    S^{\ket{\phi_{n,0}}}_i  = X_i \prod_j (Z_j)^{\Gamma^{(n,0)}_{i,j}} , \qquad 1 \leq i \leq n \, ,
 \end{equation}
where (see Eq.~\eqref{eq:adjacency-Gamma-n,0})
 \begin{equation}
   \label{egT2}
   \Gamma^{(n,0)} = \left[  \begin{array}{c|c}
    0 & -A^{\phantom{T}}  \\
    \hline
    -A^T & 0\\
 \end{array} \right] \, .
 \end{equation}
 
One method to check for $k$-uniformity is based on the stabilizer formalism. For this, as we discussed at the end of section~\ref{sec:kuni-from-codes}, one can check if in all the stabilizer generators and arbitrary products of them, identity operators appear on at most $n -k - 1$ different qudits.
In order to check this, we note that every row of the adjacency matrix $\Gamma^{(n,0)}$ corresponds to one of the stabilizer generators. 
It is easy to check that identity operators appear on at most $n-k-1$ different qudits for each of the stabilizer generators. 
Moreover, as every square submatrix of the $A$ matrix is nonsingular, it follows that any subset of up to $k$ column vectors of $\Gamma^{(n,0)}$ is linearly independent (for a proof see \cite[Chapter~1]{MacWilliams} \cite{Zahra-min-support}).
Due to this linear independence, we can conclude that multiplying stabilizer generators does increase the number of identity operators.

\subsection{The states constructed from the general method} \label{app:stabilizers-of-graph-general}
Rather than simply taking the equally weighted superposition of the codewords, we showed in the main text how to generalise the method and construct $k \uni$ states $\ket{\phi_{n,\nI}}$ by applying operators $\Op_{\nI}$.
This approach provides explicit closed form expressions for the states.
In this appendix, we show that the state $\ket{\phi_{n,\nI}}$, Eq.~\eqref{eq:state-n-nI}, corresponds to the first level of the hierarchical construction presented in part (b) of Table~\ref{tabal-3-cases}.
To do this, we first recall Eq.~\eqref{eq:state-n-nI}:
     \begin{equation}
     \begin{split}
      \ket{\phi_{n,\nI}} &= \1^{(n-\nI)} \otimes \Op_{\nI} \  \ket{\phi_{n,0}} \\
    &= \sum_{i} \ket{c^{(i)}_1, \dots , c^{(i)}_{n-\nI}} 
    \otimes Z^{-c^{(i)}_{n-(\nI-1)}} \otimes \dots \otimes Z^{-c^{(i)}_{n-(\nI-\kI)} }\otimes X^{c^{(i)}_{n-(\nI-\kI)+1}}  \otimes \dots \otimes X^{c^{(i)}_n} \ \ket{\phi_{\nI,0}} \, ,
     \end{split}
     \end{equation}
where $\ket{\phi_{\nI,0}}$ is a $\kI \uni$ state, and in this case, $\vec{c}_j = (\vec{x} \ G_{\kI,\nI})=(\vec{x}_j ,\, \vec{x}_j \, {\AI})$, and $\vec{x} $ is a vector of size $\kI$.
A given state $\ket{\phi_{n,\nI}}$ constructed using this method is the common eigenstate with eigenvalue $+1$ with respect to each of the following stabilizer generators $ S^{\ket{\phi_{n,\nI}}}_i$:
 \begin{equation} 
   \begin{cases}
    \bigotimes_{j=1}^{i-1} \1 \otimes X  \bigotimes_{j=i+1}^{k}  \1 \bigotimes_{j=1}^{n-k-\nI} X^{a_{i,j}} \bigotimes_{j=n-k-\nI+1}^{n-k-\nI+\kI} Z^{-a_{i,j}} \bigotimes_{j=n-k-\nI+\kI+1}^{n} X^{a_{i,j}}  
     & 1 \leq i \leq k \\
    \phantom{.} \\
    \bigotimes_{j=1}^{k} Z^{-a_{j,(i-k)}} \bigotimes_{j=k+1}^{i-1} \1 \otimes Z \bigotimes_{j=i+1}^{n}  \1 &  k+1 < i \leq n-\nI \\
    \phantom{.} \\
    \bigotimes_{j=1}^{k} Z^{-a_{j,(i-k)}} \bigotimes_{j=k+1}^{i-1} \1 \otimes X \bigotimes_{j=i+1}^{n-\nI+\kI} \1 \bigotimes_{j=1}^{\nI-\kI} X^{\alpha_{(i-n+\nI),j}} & n-\nI+1 \leq i \leq n-\nI+\kI \\
    \phantom{.} \\
    \bigotimes_{j=1}^k Z^{-a_{j,(i-k)}} \bigotimes_{j=k+1}^{n-\nI}\1 \bigotimes_{j=1}^{\kI} Z^{-\alpha_{j,(i-n+\nI-\kI)}} \bigotimes_{j=n-\nI+\kI+1}^{i-1} \1 \otimes Z \bigotimes_{j=i+1}^n\1  & n-\nI+\kI+1 \leq i \leq n 
     \end{cases} ,
 \end{equation} 
where matrix elements of ${\AI}$ are denoted by $\alpha_{i,j}$.
By performing local Fourier transforms $F$ on the last $n-k$ qudits, except the qudits in the interval $(n-\nI, \dots , n-\nI+\kI+1)$, the stabilizer generators can be converted into graph state form:
  \begin{equation}
  \label{kUniStab}
    S^{\ket{\phi_{n,\nI}}}_i  = X_i \prod_j (Z_j)^{\Gamma^{(n,\nI)}_{i,j}} , \qquad 1 \leq i \leq n \, ,
 \end{equation}
where 
 \begin{equation}
   \Gamma^{(n,\nI)} = 
    \left[  \begin{array}{c|c}
    0 & -A^{\phantom{T}}  \\
    \hline
    -A^T & 
    \color{blue}
     \begin{array}{c|c}    
     0 & -{\AI}\\
      \hline
     -{\AI}^T & 0
      \end{array}   
    \\
 \end{array} \right] \, .
 \end{equation}
We see that the graph representation of the state $\ket{\phi_{n,\nI}}$ is the same as what is shown in Table~\ref{tabal-3-cases}, part (b). 

Note that the stabilizer generators (\ref{kUniStab}) are all linearly independent; hence in accordance with the discussion in \cref{app:stabilizer} about expressing the density matrices of stabilizer states in terms of stabilizer operators, we conclude that $\ketbra{\phi_{n,0}}{\phi_{n,0}}$ is locally equivalent to
\begin{equation}
\label{Bip}
\rho
=\dfrac{1}{n^\dloc}
\sum_{(w_i,\ldots,w_n ) \in \mathbb{Z}_\dloc^n } 
\prod_{i=1}^n \Big(X_i \prod_j (Z_j)^{\Gamma^{(n,0)}_{i,j}} \Big)^{w_i}
.
\end{equation}
Similarly, $\ketbra{\phi_{n,\nI}}{\phi_{n,\nI}}$ is locally equivalent to
\begin{equation}
\label{nonBip}
\sigma
=\dfrac{1}{n^\dloc}
\sum_{(w_i,\ldots,w_n ) \in \mathbb{Z}_\dloc^n } 
\prod_{i=1}^n \Big(X_i \prod_j (Z_j)^{\Gamma^{(n,\nI)}_{i,j}} \Big)^{w_i}
.
\end{equation}
In the above equations, $\rho$ and $\sigma$ are local unitary equivalent to the density matrices $\ketbra{\phi_{n,0}}{\phi_{n,0}}$, and $\ketbra{\phi_{n,\nI}}{\phi_{n,\nI}}$ respectively, as we performed local Fourier gates to bring the states into graph form.

\section{Proof of Theorem~\ref{thm:general-adjacency}}\label{app:prooftheorems}
Here we give the proof of  Theorem~\ref{thm:general-adjacency}.
We first discuss the structure of the stabilizer generators of $k \uni$ states in more detail. Recall from \cref{app:stabilizers-of-graph-codes} that the stabilizer generators of state $\ket{\phi_{n,0}}$ are formed using codewords $\vec{c}_i$ and the matrix $A$, which has the property that every square submatrix is nonsingular.

\begin{lemma}
If every square submatrix of $A$ is nonsingular, any linear combination of $t$ rows has at most $t-1$ vanishing elements.
\end{lemma}

\begin{proof}
Consider any linear combination of $t$ rows of matrix $A$ with the matrix elements $a_{i,j}$:
\begin{equation} 
\vec{v}:=
\Big(
\sum_{\ell=1}^t w_{i_\ell} a_{i_{1,1}}, \ldots , \sum_{\ell=1}^t w_{i_\ell} a_{i_{1,n}} \Big)\ ,
\end{equation}  
where $w_{i_1},\ldots,w_{i_t} \in GF(\dloc)$ are non-vanishing coefficients. 
Suppose that the corresponding vector $\vec{v}$ has at least $t$ vanishing elements indexed by $j_1,\ldots,j_t$. This is equivalent to the statement that matrix $A$ restricted to the square $t \times t$ matrix formed from the elements with row indices $i_1,\ldots,i_t$ and column indices $j_1,\ldots,j_t$ is singular. This holds regardless of the number $t$ and choice of indices $j_1,\ldots,j_t$ and $i_1,\ldots,i_t$. Hence, we conclude that for a matrix $A$ for which every square submatrix of $A$ is nonsingular, any linear combination of $t$ rows has at most $t-1$ vanishing elements.
\end{proof} 

Before we show that the state described in Theorem~\ref{thm:general-adjacency} is a $k \uni$ state, we first prove a Lemma about the stabilizer operators of the $k \uni$ state $\ket{\phi_{n,0}}$:

\begin{lemma}
\label{lemmaXXX}
Suppose that the Pauli strings $S_1, S_2 ,\dots , S_n$ (Eq.~\eqref{egT}) are stabilizer generators of the $k \uni$ state $\ket{\phi_{n,0}}$. Therefore for any nonzero multi-index $(w_1,\ldots,w_n ) \in \mathbb{Z}_\dloc^n $ the product,
\begin{equation}
\label{eq1} 
S^{\ket{\phi_{n,0}}}= S_1^{w_1}\cdots S_n^{w_n} =
\prod_{i=1}^n \Big(X_i \prod_j (Z_j)^{\Gamma^{(n,0)}_{i,j}} \Big)^{w_i} \ ,
\end{equation}
has identity operators acting on at most $n-k-1$ different qudits.
\end{lemma}

\begin{proof}
As we discussed in Appendix~\ref{app:stabilizers-of-graph-codes}, the generators of the stabilizer formalism of the state $\ket{\phi_{n,0}}$ are equivalent to
  \begin{equation}
    S_i= X_i \sum_j   (Z_j)^{\Gamma^{(n,0)}_{i,j}} \ ,
  \end{equation}
for $i = \{1, \dots , n\}$, and 
 \begin{equation}
   \Gamma^{(n,0)} = \left[  \begin{array}{c|c}
    0 & -A^{\phantom{T}}  \\
    \hline
    -A^T & 0\\
 \end{array} \right] \, .
 \end{equation}
All of the stabilizer operators can be expressed as $S^{\ket{\phi_{n,0}}} = S_1^{w_1}\cdots S_n^{w_n}$, where $w_1 ,\dots , w_n \in \mathbb{Z}_\dloc$.
Here we want to show that the number of the identity matrices in $S$ is at most $n-k-1$.

Consider any non-zero multi-index $(w_1,\ldots,w_n ) \in \mathbb{Z}_\dloc^n $. Since $(w_1,\ldots,w_n )\neq 0$, there are two possibilities, either $(w_1,\ldots,w_k )\neq 0$ or $(w_{k+1},\ldots,w_n )\neq 0$ (or both). Firstly, suppose that $(w_1,\ldots,w_k )\neq 0$, and denote by $t$ the number of non-vanishing elements in the tuple $(w_1,\ldots,w_k )$, for example, in $(\underbrace{1,\dots , 1}_{t}, \underbrace{0, \dots , 0}_{k-t})$ the first $t$ elements are non-vanishing.
In that case, the operators $X_i$ will appear on at least $t$ positions out of the first $k$ positions in 
\begin{equation}
\label{eq5}
\prod_{i=1}^k \Big(X_i \sum_j (Z_j)^{\Gamma^{(n,0)}_{i,j}} \Big)^{w_i}.
\end{equation}
Recall that the form of the matrix $\Gamma^{(n,0)}$ is related to the matrix $A$ by \cref{egT2}. Since the matrix $A$ is nonsingular, the operators $Z_i$ will appear on at least $n-k-t+1$ positions out of the last $n-k$ positions in \cref{eq5}. 
Note that in 
\begin{equation}
\label{eq6}
\prod_{i=k+1}^n \Big(X_i \sum_j (Z_j)^{\Gamma^{(n,0)}_{i,j}} \Big)^{w_i},
\end{equation}
there are no $X_i$ operators on the first $k$ positions and $Z_i$ operators on the last $n-k$ positions. 
Now we check the number of identity matrices in Eq.~\eqref{eq1}.
Obviously, (\ref{eq1})$=$(\ref{eq5})$\cdot$(\ref{eq6}), and hence the operator in \cref{eq1} has non-identity elements on at least $t$ positions out of the first $k$ positions, and on at least $n-k-t+1$ positions among the last $n-k$ positions. In total, the number of non-identity elements in \cref{eq1} is greater than or equal to $n-k+1$. Note that $n-k+1>k+1$.

Secondly, assume that $(w_{k+1},\ldots,w_n )\neq 0$ and denote by $t$ the number of non-vanishing elements in the tuple $(w_{k+1},\ldots,w_n )$. Similarly to the previous calculations, one may conclude that the number of non-identity elements in \cref{eq1} on the first $k$ positions is greater than or equal to $k-t+1$, while the number of non-identity elements in \cref{eq1} on the last $n-k$ positions is greater than or equal to $t$. In total, the number of non-identity elements in \cref{eq1} is greater than or equal to $k+1$, which proves the statement. 
\end{proof}

We are now ready to prove Theorem~\ref{thm:general-adjacency}. 
For this we use \cref{lemmaXXX} and show that the number of identity matrices in the stabilizer operators corresponding to the state $\ket{\phi}$ described in the theorem, i.e.,
 \begin{equation}
   S^{\ket{\phi}}= S_1^{w_1} \cdots S_n^{w_n} =
\prod_{i=1}^n \Big(X_i \sum_j (Z_j)^{\Gamma^{n}_{i,j}} \Big)^{w_i} \ ,
  \end{equation}
 with (recall Eq.~\eqref{eq:general-adjacency}) 
    \begin{equation}
       \Gamma^{n} = \left[  \begin{array}{c|c}
    0 & -A^{\phantom{T}}  \\
    \hline
    -A^T & \color{blue} B \\
 \end{array} \right] \ ,
    \end{equation}
is at most $n-k-1$.

\begin{lemma}
The state $\ket{\phi}$ corresponding to the adjacency matrix $\Gamma^n$ is $k \uni$, because for any nonzero multi-index $(w_1,\ldots,w_n ) \in \mathbb{Z}_\dloc^n $, the combination of Pauli strings
\begin{equation}
\label{eq9}
\prod_{i=1}^n \Big(X_i \sum_j (Z_j)^{\Gamma^{n}_{i,j}} \Big)^{w_i} \ ,
\end{equation}
has identity elements acting on at most $n-k-1$ different qudits.
\end{lemma}

\begin{proof}
Consider the nonzero multi-index $(w_1,\ldots,w_n ) \in \mathbb{Z}_\dloc^n $. Similarly as in the proof of \cref{lemmaXXX}, we shall consider two cases, either $(w_{k+1},\ldots,w_n )= 0$ and hence $(w_1,\ldots,w_k )\neq 0$, or $(w_{k+1},\ldots,w_n )\neq 0$. 

In the first case, the proof absolutely agrees with the proof of \cref{lemmaXXX}, since $\Gamma^{n}_{i,j} \equiv \Gamma^{(n,0)}_{i,j}$ for $i\leq k$. Consider now the second case, i.e., suppose that $(w_{k+1},\ldots,w_n )\neq 0$, and denote by $t$ the number of non-vanishing elements in the tuple $(w_{k+1},\ldots,w_n )$. Note that the operators $X_i$ will appear on at least $t$ positions out of the last $n-k$ positions in 
\begin{equation}
\label{eq7}
\prod_{i=k+1}^n \Big(X_i \sum_j (Z_j)^{\Gamma^{n}_{i,j}} \Big)^{w_i}.
\end{equation}
Note that the matrix $\Gamma^{n}_{i,j} \equiv \Gamma^{(n,0)}_{i,j} $ for all $j\leq k$, and it is related to the matrix $A$ by (\ref{eq:adjacency-Gamma-n,nI,nII,etc}). 
Since the matrix $A$ is nonsingular, the operators $Z_i$ will appear on at least $k-t+1$ positions out of the first $k$ positions in \cref{eq7}. 
Note that in 
\begin{equation}
\label{eq8}
\prod_{i=1}^k \Big(X_i \sum_j (Z_j)^{\Gamma^{n}_{i,j}} \Big)^{w_i},
\end{equation}
there are no $Z_i$ operators on the first $k$ positions and $X_i$ operators on the last $n-k$ positions. Obviously, (\ref{eq9})$=$(\ref{eq7})$\cdot$(\ref{eq8}), and hence the operator \cref{eq9} has non-identity elements on at least $k-t+1$ positions out of the first $k$ positions, and on at least $t$ positions among the last $n-k$ positions. In total, the number of non-identity elements in \cref{eq9} is greater than or equal to $k+1$, which finishes the proof. 
\end{proof}

Note that the proof is valid for the adjacency matrix of the form
 \begin{equation}
   \Gamma^n= \left[  \begin{array}{c|c}
    0 & -A^{\phantom{T}}  \\
    \hline
    -A^T & \color{blue} B\\
 \end{array} \right] \, ,
 \end{equation}
where the only conditions we considered are that $A$ is a nonsingular matrix, and $B$ is an arbitrary adjacency matrix. 
But in this paper, we focus on the hierarchical structure, therefore we choose a specific form for the $B$ matrix.
For example, in the first level of iteration  we have
  \begin{equation}
   B=  \left[
    \color{blue}
     \begin{array}{c|c}    
     0 & 0\\
      \hline
     0 & 
     \begin{array}{c|c}    
     0 & -{\AI}\\
      \hline
     -{\AI}^T & 0
      \end{array}   
      \end{array} \right] \ ,
 \end{equation}
and in the second level of iteration, it is
   \begin{equation}
   B '=  \left[
       \color{blue}
     \begin{array}{c|c}    
     0 & 0\\
      \hline
     0 & 
     \begin{array}{c|c}    
     0 & -{\AI}\\
      \hline
     -{\AI}^T & 
            \color{red}
         \begin{array}{c|c}    
     0 & 0\\
      \hline
     0 & 
     \begin{array}{c|c}    
     0 & -{\AII}\\
      \hline
     -{\AII}^T &  0
      \end{array}   
      \end{array}
      \end{array}   

      \end{array} \right] \ .
 \end{equation}
 And in general, the hierarchical method constructs $k \uni$ states at every level of iteration, i.e., the adjacency matrix

  \begin{align}
\Gamma^{(n,\nI,\nII,\dots)}
   =  
    \left[  \begin{array}{c|c}
    0 & -A^{\phantom{T}}  \\
    \hline
    -A^T & 
    \color{blue}
     \begin{array}{c|c}    
     0 & 0\\
      \hline
     0 & 
     \begin{array}{c|c}    
     0 & -{\AI}\\
      \hline
     -{\AI}^T & 
            \color{red}
         \begin{array}{c|c}    
     0 & 0\\
      \hline
     0 & 
     \begin{array}{c|c}    
     0 & -{\AII}\\
      \hline
     -{\AII}^T &  \color{violet}{\ddots}
      \end{array}   
      \end{array}
      \end{array}   
      \end{array}
    \\
 \end{array} \right] \ ,
  \end{align} 
corresponds to the states $\ket{\phi_{n,\nI,\nII,\dots}}$ that are $k \uni$. 
 
\section{Inequivalence under SLOCC}\label{app:inequivalence-SLOCC}
It is well known that the number of product states needed to specify a pure state is an upper bound to the rank of all possible reduced density matrices.
It is also known that this number cannot be increased by SLOCC operations. 
This implies that, for a $k$-UNI state $\ket{\phi_{(n,0)}}$, and for any subset $ S\subset \{1,\dots ,n\}$, one has
 \begin{equation} 
  \text{rank} (\rho_S) \leq q ^k \ ,
 \end{equation}
where $\rho_S=\Tr_{S^c} \ketbra{\phi_{(n,0)}}{\phi_{(n,0)}}$. 
This means that all the reductions up to $k$ parties of the state $\ket{\phi_{(n,0)}}$ are maximally mixed, i.e., it is a $k \uni$ state.
However, if one can find a pure state that is a  $k \uni$ state and in addition to this, there exists at least one subset of size larger than $k$ parties such that the reduced density matrix are maximally mixed, then these two $k \uni$ states cannot be converted into one another via SLOCC. 

In this section, we show that some reduced states $\rho_S =\Tr_{S^c} \ketbra{\phi_{n,\nI}}{\phi_{n,\nI}}$,  for subsystem size $|S|>k$, are maximally mixed.
More precisely, we show that $\rho_S$ for a subset $S$ comprised of two parts $S_1$ and $S_2$, is maximally mixed, where $S_1$ has size $|S_1|=k$ and is contained entirely in the support of the first $n-\nI$ qudits, i.e., the subset $\{1, \dots, n-\nI \}$, and $S_2$ (with size $|S_2|=\kI$) is contained in the last $\nI$ qudits, i.e, the subset $\{n-\nI+1, \dots, n \}$.
Concretely, we want to show 
  \begin{equation}\label{eqapp:reduced-dencity-SLOCC}
  \begin{split}
     \rho_S &=\Tr_{S^c} \ketbra{\phi_{n,\nI}}{\phi_{n,\nI}} \\
     & =\Tr_{S_1^c} \Tr_{S_2^c} \ketbra{\phi_{n,\nI}}{\phi_{n,\nI}} \\
     & \propto \1 \ .
   \end{split}
  \end{equation}

To proceed further, let us first review the structure of the state $\ket{\phi_{n,\nI}}$.
This state is obtained by applying operator  $\Op_{\nI}$ on the state $\ket{\phi_{n,0}}$, Eq.~\eqref{eq:state-n-nI}, and can be expanded as
   \begin{equation}\label{eqapp:general-state-n-nq}
   \begin{split}
    \ket{\phi_{n,\nI}} 
     &= \1^{(n-\nI)} \otimes \Op_{\nI} \ \ket{\phi_{n,0}} \\
    & = \sum_{i} \ket{c^{(i)}_1, \dots , c^{(i)}_{n-\nI}} 
    \otimes Z^{-c^{(i)}_{n-(\nI-1)}} \otimes \dots \otimes Z^{-c^{(i)}_{n-(\nI-\kI)} }\otimes X^{c^{(i)}_{n-(\nI-\kI)+1}}  \otimes \dots \otimes X^{c^{(i)}_n} \ \ket{\phi_{\nI,0}} \\
    &= \sum_{i} \ket{c^{(i)}_1, \dots , c^{(i)}_{n-\nI}}  \otimes \ket{\psi_{i}^{\nI}} \ ,
    \end{split}
   \end{equation}  
where as above we denote 
 \begin{equation}\label{eqapp:basis}
   \ket{\psi_{i}^{\nI}} \coloneqq Z^{-c^{(i)}_{n-(\nI-1)}} \otimes \dots \otimes Z^{-c^{(i)}_{n-(\nI-\kI)} }\otimes X^{c^{(i)}_{n-(\nI-\kI)+1}}  \otimes \dots \otimes X^{c^{(i)}_n} \ \ket{\phi_{\nI,0}} \ .
 \end{equation}
The state $\ket{\phi_{\nI,0}}$ is a $\kI \uni$ state constructed from classical codes using the method presented in Eq.~\eqref{eq:state-n,0}.
For every $i$, the state $\ket{\psi_{i}^{\nI}}$ is also a $\kI \uni$ state, as acting with local unitaries on a given state does not change the entanglement properties.
Moreover, it is proven in \cite[Lemma~1]{Zahra-non-min-support} that the states $\ket{\psi_{i}^{\nI}}$ for $i \in \{1,\dots ,\dloc^{\nI}\}$ form a complete orthonormal basis of $\kI \uni$ states with $\nI$ qudits, i.e., $\braket{\psi_{i}^{\nI}}{\psi_{i'}^{\nI}}=\delta_{i,i'}$.

To prove that the reduced density matrix $\sigma_S$ is maximally mixed for the given subset of size $|S|=k+\kI$, i.e, to show that Eq.~\eqref{eqapp:reduced-dencity-SLOCC} holds, we check two different cases: (i) $k \leq \nI$ and (ii) $k > \nI$.

  \begin{itemize}
    \item[case (i)] We first consider the case $k \leq \nI$, where we have
    \begin{equation}
       \begin{split}
          \rho_S
     & =\Tr_{S_1^c} \Tr_{S_2^c} \ketbra{\phi_{n,\nI}}{\phi_{n,\nI}} \\
     & = \Tr_{S_1^c} \sum_{i,i'}  \ketbra{c_1^{(i)},\dots , c_{n-\nI}^{(i)}}{c_1^{(i')},\dots , c_{n-\nI}^{(i')}}
     \otimes \Tr_{S_2^c} \ketbra{\psi_i^{\nI}}{\psi_{i'}^{\nI}} \ .
        \end{split}
    \end{equation}
We know that the states $\ket{\psi_{i}^{\nI}}$, Eq.~\eqref{eqapp:basis}, are a complete basis of $\kI \uni$ states.
When $k = \nI$ the set $i = \{1,\dots ,\dloc^k\}$ contains the entire basis, and otherwise in the case of having $k < \nI$, we have only part of the basis.
It is obvious that in both cases, $\Tr_{S_2^c} \ketbra{\psi_i^{\nI}}{\psi_i^{\nI}} = \delta_{i,i'} \ \1_{\kI}$. 
Therefore, the reduced density matrix simplifies to
 \begin{equation}
     \begin{split}
   \rho_S  &= \Tr_{S_1^c} \sum_{i,i'}  \ketbra{c_1^{(i)},\dots , c_{n-\nI}^{(i)}}{c_1^{(i')},\dots , c_{n-\nI}^{(i')}}
     \otimes \1_{\kI} \ \delta_{i,i'} \\
     & = \1_k \otimes \1_{\kI}\ .
        \end{split}
 \end{equation}
where we used the fact that $\vec{c}_i =(\vec{x}_i \, G_{k\times n})=(\vec{x}_i ,\ \vec{x}_i A )$ are codewords of a suitable classical code such that the number of free indices in the code is equal to $k$.

    \item[case (ii)] In the case where $k > \nI$, there can be a repetition in the basis $\ket{\psi_i^{\nI}}$, because the number of the states that form a complete orthonormal basis in a Hilbert space $\hiH(\nI,\dloc) \coloneqq \mathbb{C}_\dloc^{\otimes \nI}$, is equal to $\dloc^{\nI}$, while  $i \in \{1,\dots, \dloc^k\}$.
Therefore, to proceed with the proof, we consider the known conditions on the codewords $\vec{c}_i=(c_1^{(i)},\dots,c_n^{(i)})$ that are used to construct the state $\ket{\phi_{n,\nI}}$ and the orthonormal basis $\ket{\psi_{i}^{\nI}}$.

We can have two conditions depending on the exponents of the $X$ and $Z$ operators in Eq.~\eqref{eqapp:basis}.
One corresponds to the case that the two sequences $c_{n-(\nI-1)}^{(i)}, \dots, c_{n}^{(i)}$ and $c_{n-(\nI-1)}^{(i')}, \dots, c_{n}^{(i')}$ differ in at least one position.
In this case we get $\Tr_{S_{1}^{c}} \ketbra{\psi_{i}^{\nI}}{\psi_{i'}^{\nI}} = \delta_{ii'} \ \1_{\kI}$.
Therefore, as for case (i), we have $\rho_S = \1_k \otimes \1_{\kI}$.

In the second condition, the two sequences $c_{n-(\nI-1)}^{(i)}, \dots, c_{n}^{(i)}$ and $c_{n-(\nI-1)}^{(i')}, \dots, c_{n}^{(i')}$ that are part of the codewords $\vec{c}_i$ and $\vec{c}_{i'}$, do not differ.
Therefore, these parts of the codewords, with the size $\nI$, have overlap.

We know that, based on the properties of classical linear codes (see Appendix \ref{app:classical linear codes}), in general, every two codewords $\vec{c}_i$ and $\vec{c}_{i'}$ with $i \neq i'$ have a distance at least equal to the Singleton bound \eqref{eqapp:Singleton}, i.e., $\dH=n-k+1$.
As sequences of these codewords, with size $\nI$, have overlap, the remaining part is orthogonal
 \begin{equation}
   \Tr_{S_1^c} \sum_{i,i'}  \ketbra{c_1^{(i)},\dots , c_{n-\nI}^{(i)}}{c_1^{(i')},\dots , c_{n-\nI}^{(i')}} =  \delta_{i,i'} \Tr_{S_1^c} \sum_{i}  \ketbra{c_1^{(i)},\dots , c_{n-\nI}^{(i)}}{c_1^{(i)},\dots , c_{n-\nI}^{(i)}}  \ .
 \end{equation}
In other words, we used the fact that $\dH-\nI \geq |S_{1}^{c}| = n-\nI-k$, and considering this, there cannot be any overlap with the sequences $c_1^{(i)},\dots , c_{n-\nI}^{(i)}$ and $c_1^{(i')},\dots , c_{n-\nI}^{(i')}$ of the codewords.
Putting all this together, the reduced density matrix $\rho_S$ simplifies to
    \begin{equation}
       \begin{split}
          \rho_S
      &= \Tr_{S_1^c} \sum_{i}  \ketbra{c_1^{(i)},\dots , c_{n-\nI}^{(i)}}{c_1^{(i)},\dots , c_{n-\nI}^{(i)}}
     \otimes \Tr_{S_2^c} \ketbra{\psi_i^{\nI}}{\psi_{i}^{\nI}} \\
     & = \1_k \otimes \1_{\kI} \ ,
        \end{split}
    \end{equation}
in which we used the fact that every linear code contains $k$ free indices.   
  \end{itemize}
  
Now, we can conclude that for both cases $k \leq \nI$ and $k > \nI$, the reduced density matrix $\sigma_S$, for $S$ given by the union of $S_1$ and $S_2$, is maximally mixed, where $S_1 \subseteq \{1, \dots, n-\nI\}$ with $|S_1|=k$ and $S_2 \subseteq \{n-\nI+1, \dots, n\}$ with $|S_2|=\kI$.
Obviously, this can only be achieved if $|S|=k+\kI \leq n/2$. 

Here, we have only proven that the two sets of $k \uni$ states, $\ket{\phi_{(n,0)}}$ and $\ket{\phi_{(n,\nI)}}$, 
belong to different SLOCC classes, we leave the proof for the more general cases to future, in which we will show that all the $k \uni$ states $\ket{\phi_{(n,\nI)}}$ and $\ket{\phi_{(n,{\nI}')}}$ 
with $\nI \neq {\nI}'$ as well as those $k \uni$ states belonging to different hierarchy levels.

\section{Two inequivalent types of AME states under SLOCC}\label{app:inequivalence-SLOCC-AME}
In this section, we examine the problem of SLOCC-discrimination of different AME states with the same number of qudits $n$ and local dimension $\dloc$.
Here, we consider two states $\ket{\phi_{n,0}}$ and $\ket{\phi_{n,{2}}}$ of odd numbers of qudits $n=2k+1$ defined by Eq.~\eqref{eq:state-n-nI}.
We demonstrated in \cref{app:prooftheorems}, that if state $\ket{\phi_{n,0}}$ is an AME state, then the second state (i.e., the state $\ket{\phi_{n,{2}}}$) is also an AME state.
We shall see that they are not SLOCC-equivalent. 
Observe that for each subsystem, the reduced density matrices of both states always have the same rank.
Therefore the method used in \cref{app:inequivalence-SLOCC} for verification of SLOCC-equivalence between two $k \uni$ states with $k < n/2$, cannot be applied anymore.
In the following we explain the proof.

In order to simplify notation, we renormalize both states, $\ket{\phi_{n,0}}$ and $\ket{\phi_{n,{2}}}$, by the same factor $\sqrt{\dloc^k}$.

Firstly, we consider again the closed formula for the state $\ket{\phi_{n,0}}$ and its reduction $\rho_{S} ( \phi_{n,{0}}) = \ketbra{\phi_{n,{0}}}{\phi_{n,{0}}}$, to the subsystem $S=\{1,\ldots,k+1 \}$ that includes the first $(k+1)$ qudits. 
According to \cref{eq:kunimin-general}, the state $\ket{\phi_{n,0}}$ can be written as follows:
 \begin{equation}\label{eq:1}
   \ket{\phi_{n,0}}
   =\sum_{\vec{x}_i \in [ \dloc ]^{k}} \ket{\vec{x}_i,\, \vec{x}_i A} \ ,
 \end{equation}
where $[q] \coloneqq (0; \dots, q-1)$, and $A$ is $k \times (n-k)$ matrix with elements in $ GF(\dloc)$. 
We separate qudits $k$ and $k+1$ in the notation, thus \cref{eq:1} reads
 \begin{equation}\label{eq:2}
   \ket{\phi_{n,0}}=
   \sum_{\substack{ \vec{x}_i \in [ \dloc ]^{k-1} \\ x_j \in [\dloc ]}} 
   \underbrace{\ket{\vec{x}_i}}_{k-1} 
   \otimes 
     \underbrace{\ket{x_j}}_{1} 
        \otimes 
\underbrace{\ket{\vec{x}_i A_1+ a x_j}}_{1} 
   \otimes 
\underbrace{\ket{\vec{x}_i A_2+ x_j A_3}}_{n-k-1} 
 ,
 \end{equation}
where the matrix $A$ is separated into four submatrices
\begin{equation}\label{division}
A=
  \begin{bmatrix}
    \begin{array}{c|c}
  A_1 & A_2 \\
  \hline
 a & A_3
    \end{array}
  \end{bmatrix}
\end{equation}
where $A_1,A_2,A_3$ are of size $(k-1)\times 1$, $(k-1)\times (n-k-1)$, and $1\times (n-k-1)$ respectively, while $a \in GF (\dloc )$. 
Since $\ket{\phi_{n,0}}$ is an AME($2k+1$,$\dloc$) state, the states $\ket{\vec{x}_i A_2+ x_j A_3}$ are linearly independent for different values of $\vec{x}_i \in [ \dloc ]^{k-1},  x_j \in [\dloc ]$. 
In fact, the set of states $\{\ket{\vec{x}_i A_2+ x_j A_3}\}_{\vec{x}_i \in [ \dloc ]^{k-1},  x_j \in [\dloc ]} \in \hiH(k,\dloc)$ forms a basis in $\hiH(k,\dloc) \coloneqq \mathbb{C}_\dloc^{\otimes k}$. 
Therefore the reduced density matrix $\rho_{S} ( \phi_{n,{0}})$, with $S$ corresponding to the  first $k+1$ qudits reads
\begin{equation}\label{eq:6}
\rho_{S} \Big ( \phi_{n,{0}} \Big)   
= 
\sum_{\substack{ \vec{x}_i \in [ \dloc ]^{k-1} \\ x_j \in [\dloc ]}} 
\ket{\vec{x}_i}
\ket{x_j}
\ket{\vec{x}_i A_1+ a x_j}
\bra{\vec{x}_i}
\bra{x_j}
\bra{\vec{x}_i A_1+ a x_j}.
\end{equation}

Secondly, we study the state $\ket{\phi_{n,2}}$ and its reduction $\rho_{S} ( \phi_{n,{2}}) $ to the subsystem $S=\{1,\ldots,k+1 \}$ containing the first $(k+1)$ qudits. 
Note that the maximally entangled state of two qudits with the local dimension $\dloc$ is locally equivalent to the state
\begin{equation}
   \ket{\phi_{2,0}}=\dfrac{1}{\sqrt{\dloc}}\sum_{\ell \in [\dloc ]} \ket{\ell,\ell}.
\end{equation}
The state $\ket{\phi_{n,2}}$ is constructed from $\ket{\phi_{n,0}}$ and $\ket{\phi_{2,0}}$ via the non-local operator $\Op_{2}$ acting on the last two qudits of $\ket{\phi_{n,0}}$, as indicated in \cref{eq:Operator,eq:state-n-nI}. 
Because of the properties that $k \uni$ states have, after any permutation of its subsystem, without loss of generality, we can change the subspace of the action of the operator $\Op_{2}$ onto subsystem $\{k,k+1\}$.
 In this way, the state $\ket{\phi_{n,2}}$ reads
 \begin{equation}\label{eq:4}
   \ket{\phi_{n,2}}=
      \dfrac{1}{\sqrt{\dloc}}
   \sum_{\substack{ \vec{x}_i \in [ \dloc ]^{k-1} \\ x_j,\ell \in [\dloc ]}} 
   \underbrace{\ket{\vec{x}_i}}_{k-1} 
   \otimes 
     \underbrace{X^{x_j}\ket{\ell}}_{1} 
        \otimes 
\underbrace{Z^{\vec{x}_i A_1+ a x_j} \ket{\ell}}_{1} 
   \otimes 
\underbrace{\ket{\vec{x}_i A_2+ x_j A_3}}_{n-k-1} .
 \end{equation}
As we noticed before, states $\{\ket{\vec{x}_i A_2+ x_j A_3}\}_{\vec{x}_i \in [ \dloc ]^{k-1},  x_j \in [\dloc ]} \in \hiH(k,\dloc)$ form a basis of $\mathbb{C}_\dloc^{\otimes k}$. 
Therefore the state $\ket{\phi_{n,2}}$ reduced to the subsystem $S$ reads
\begin{equation}\label{eq:5}
\rho_{S} \Big( \phi_{n,{2}}\Big)   
= \dfrac{1}{\dloc}
\sum_{\substack{ \vec{x}_i \in [ \dloc ]^{k-1} \\ x_j \in [\dloc ]}} 
\sum_{\ell,\ell' \in [\dloc]}
\omega^{(\vec{x}_i A_1+ a x_j ) (\ell-\ell')}
\ket{\vec{x}_i}
\ket{ x_j +\ell}
\ket{\ell}
\bra{\vec{x}_i}
\bra{ x_j +\ell'}
\bra{\ell'},
\end{equation} 
where $\omega$ is a $\dloc$th root of unity. 
Note the difference in the last two positions in
$\rho_{S} ( \phi_{n,{0}} ) $ and $\rho_{S} ( \phi_{n,{2}} ) $ given in \cref{eq:6} and \cref{eq:5}, respectively. 
We shall show the following. 
\begin{lemma}
\label{lemmaXX}
Both reduced density matrices are LU-equivalent by the following transformation:
\begin{equation}
U=\underbrace{\1 \otimes\cdots \otimes \1}_{k-1} \otimes W \otimes V,
\end{equation}
where unitary matrices $W=[w_{ij}]_{i,j=1}^\dloc $ and $V=[v_{ij}]_{i,j=1}^\dloc $ are defined by $w_{ij}:=\frac{1}{\dloc} \omega^{-a ij+j^2 T}$ and $v_{ij}:=\frac{1}{\dloc}\omega^{ij-j^2 T}$ 
where $\omega$ is a $\dloc$-th root of unity, while $T$ is an element in $GF(\dloc )$ such that $2T=a$. 
\end{lemma}

\begin{proof}
Observe that
\begin{align}
\nonumber
\rho_{S} \Big( \phi_{n,{0}}\Big)   
&= \nonumber
\sum_{\vec{x}_i \in [ \dloc ]^{k-1} } 
\sum_{\ell \in [ \dloc ] } 
\ket{\vec{x}_i}
\ket{\ell}
\ket{\vec{x}_i A_1+ a \ell}
\bra{\vec{x}_i}
\bra{\ell}
\bra{\vec{x}_i A_1+ a \ell} 
\xmapsto{\1^{\otimes (k-1)} \otimes W \otimes V} 
\\
&\sum_{\vec{x}_i \in [ \dloc ]^{k-1}} 
\ket{\vec{x}_i}
\bra{\vec{x}_i}
\Big(\dfrac{1}{\dloc^2}
\sum_{\substack{m,m', \\  k,k' \in [\dloc ]}} 
\sum_{\ell \in [\dloc ]} 
\omega^{\big(
-a\ell (m- m') +T (m^2-{m'}^2) +(\vec{x}_i A_1+ a \ell)(k- k') -T(k^2-{k'}^2)
\big)} 
\ket{m,k}\bra{m',k'} \Big) .
\label{eq:7}
\end{align}
\noindent
In the above expression, we change the summation indices $m,m'$ to $j:=m-k$ and $j':=m'-k'$ respectively. 
In this way, the right-hand side of \cref{eq:7} reads
\begin{align}
\nonumber
\text{(\ref{eq:7})}=&
\sum_{\vec{x}_i \in [ \dloc ]^{k-1}} 
\ket{\vec{x}_i}
\bra{\vec{x}_i}
\\&\nonumber
\sum_{\substack{k,k', \\  j,j' \in [\dloc ]}} 
\Bigg(\dfrac{1}{\dloc^2}
\sum_{\ell \in [\dloc ]} 
\omega^{\big(
-a\ell (j+k- j'-k') +T ((j+k)^2-{(j'+k')}^2) +(\vec{x}_i A_1+ a \ell)(k- k') -T(k^2-{k'}^2)
\big)}
\Bigg)
\ket{k+j,k}\bra{k'+j',k'}
\\&\nonumber
=\sum_{\vec{x}_i \in [ \dloc ]^{k-1}} 
\ket{\vec{x}_i}
\bra{\vec{x}_i}
\sum_{\substack{k,k', \\  j,j' \in [\dloc ]}} 
\Bigg(\dfrac{1}{\dloc^2}
\underbrace{\sum_{\ell \in [\dloc ]} 
\omega^{-a \ell (j-j')}}_{=\dloc \delta_{j j'}}
\omega^{T (j^2-{j'}^2)+ 2T (kj-k'j')}
\omega^{(\vec{x}_i A_1)(k- k')} 
\Bigg)
\ket{k+j,k}\bra{k'+j',k'}
\\&\nonumber
=\sum_{\vec{x}_i \in [ \dloc ]^{k-1}} 
\ket{\vec{x}_i}
\bra{\vec{x}_i}
\sum_{\substack{k,k', \\  j,j' \in [\dloc ]}} 
\Bigg(\delta_{j j'}\dfrac{1}{\dloc} 
\omega^{(2Tj+\vec{x}_i A_1) (k-k')} 
\Bigg)
\ket{k+j,k}\bra{k'+j',k'}
\\&\nonumber
=\sum_{\vec{x}_i \in [ \dloc ]^{k-1}} 
\ket{\vec{x}_i}
\bra{\vec{x}_i}
\sum_{k,k' \in [\dloc ]} 
\sum_{j\in [\dloc ]} 
\Bigg(\dfrac{1}{\dloc} 
\omega^{(aj+\vec{x}_i A_1) (k-k')} 
\Bigg)
\ket{k+j,k}\bra{k'+j,k'},
\end{align}
which is equal to $\rho_{S} ( \phi_{n,{2}}) $ (compare with \cref{eq:5}).
\end{proof}

Now, we recall the notion of a \textit{monomial matrix}, and the \emph{support of a state}.
The \textit{support} of a state $\ket{\psi}$ is the number of non-zero coefficients when $\ket{\psi}$  is written in the computational basis.

\begin{definition}
A unitary matrix $M$ is called a \textit{unitary monomial matrix} if one of the following holds:
\begin{enumerate}
\item $M$ has exactly one nonzero entry in each row and each column,
\item $M$ is a product of a permutation and diagonal matrix,
\item $M$ does not change the support of any quantum state.
\end{enumerate}
\end{definition}

\begin{proof}[Proof of Theorem \ref{TheoremAppF}]
We show the statement by contradiction. 
Assume that states $\ket{\phi_{n,0}}$ and $\ket{\phi_{n,2}}$ are LU-equivalent by some unitary matrices: 
\begin{equation*}
U_1 \otimes \cdots \otimes U_{k+1} \otimes U_{k+2} \otimes \cdots \otimes U_{n}
\end{equation*}
We shall keep "$\otimes$" in the notation in order to distinguish it from matrix multiplication. 
Since the (partial) trace is invariant under cycling permutations, one can show that:
\begin{equation}
\rho_{S} ( \phi_{n,{0}}) 
=
\Big( U_1 \otimes \cdots \otimes U_{k+1} \Big)
\rho_{S} ( \phi_{n,{2}}) 
\Big( U_1^\dagger \otimes \cdots \otimes U_{k+1}^\dagger \Big)
\end{equation}
\noindent
Hence, the operator $ \widetilde{U}:=  ( U_1 \otimes \cdots \otimes U_{k+1} )$ provides the local equivalence between $\rho_{S} ( \phi_{n,{0}}) $ and $\rho_{S} ( \phi_{n,{2}}) $. 
Notice that in \cref{lemmaXX}, we pointed out that the local matrices $U=\1^{\otimes k-1} \otimes  W \otimes V$ provide the LU-equivalence between $\rho_{S} ( \phi_{n,{0}}) $ and $\rho_{S} ( \phi_{n,{2}}) $. 
Therefore $\widetilde{U} U $ provides an LU-equivalence between $\rho_{S} ( \phi_{n,{0}}) $ and itself, i.e.
\begin{equation}
\rho_{S} ( \phi_{n,{0}}) 
=
\widetilde{U} U
\rho_{S} ( \phi_{n,{0}})
U^\dagger \widetilde{U}^\dagger
. 
\end{equation}
Note that the operator $\widetilde{U} U $ has the following form
\begin{equation}
\label{monomial}
\widetilde{U} U = U_1 \otimes \cdots \otimes U_{k-1} \otimes U_{k} W \otimes U_{k+1} V.
\end{equation}
According to \cite[Proposition 2]{BurchardtRaissi20}, all matrices in \cref{monomial}, i.e. $U_1,\ldots  , U_{k-1} ,U_{k} W ,U_{k+1} V$ are monomial matrices. 

We shall see that such a restriction on matrices $U_1,\ldots  , U_{k-1} ,\ U_{k} W ,\ U_{k+1} V$ leads to a contradiction. 
To sum up the discussion so far, LU-equivalence between states $\ket{\phi_{n,0}}$ and $\ket{\phi_{n,2}}$ has the following form:
\begin{equation*}
M_1 \otimes \cdots \otimes M_{k-1} 
\otimes 
M_{k} W \otimes M_{k+1} V
\otimes
U_{k+2} \otimes \cdots \otimes U_{n}
\end{equation*}
where $U_i$ are arbitrary unitary matrices, $W,V$ are defined in \cref{lemmaXX}, while $M_i$ are products of diagonal and permutation matrices. 
From the form of \cref{eq:2}, we have
\begin{equation}
\label{star}
   \ket{\phi_{n,2}}=
   \sum_{\substack{ \vec{x}_i \in [ \dloc ]^{k-1} \\ x_j \in [\dloc ]}} 
   \underbrace{\vec{M}\ket{\vec{x}_i}}_{k-1} 
   \otimes 
     \underbrace{B_{\vec{x}_i,x_j}}_{2} 
   \otimes 
\underbrace{ \vec{U} \ket{\vec{x}_i A_2+ x_j A_3}}_{n-k-1} ,
\end{equation}
where $\vec{M}=M_1 \otimes \cdots \otimes M_{k-1} $, and $\vec{U}=U_{k+2} \otimes \cdots \otimes U_{n} $, and 
\begin{equation}
B_{\vec{x}_i,x_j}:= 
M_{k} W \ket{x_j} 
\otimes 
M_{k+1} V \ket{\vec{x}_i A_1+ a x_j}.
\end{equation}
Since the matrix $A$ is nonsingular, its $k\times k $ submatrix $\left[\begin{smallmatrix}A_2\\  \hline A_3\end{smallmatrix}\right]$ is also non-singular. As a consequence the following $n-k-1=k$ states: 
$
\ket{\vec{x}_i A_2+ x_j A_3}
$ 
are linearly independent for different values of multi-indices $\vec{x}_i, x_j \in [\dloc^k ]$. Similarly, their unitary transformation 
$
 \vec{U} \ket{\vec{x}_i A_2+ x_j A_3}
$ 
remains linearly independent. Furthermore, from the form of matrices $V$ and $W$, the support of the state $B_{\vec{x}_i,x_j}$ equals $\dloc^2$ for any multi-index $\vec{x}_i, x_j \in [\dloc^k ]$. Therefore, we conclude that the support of a state (\ref{star}) equals at least $\dloc^{k+2}$.

On the other hand, from \cref{eq:4}, we have
\begin{equation}
\label{comp2}
   \ket{\phi_{n,2}}=
   \sum_{\substack{ \vec{x}_i \in [ \dloc ]^{k-1} \\ x_j \in [\dloc ]}} 
   \underbrace{\ket{\vec{x}_i}}_{k-1} 
   \otimes 
     \underbrace{C_{\vec{x}_i,x_j}}_{2} 
   \otimes 
\underbrace{\ket{\vec{x}_i A_2+ x_j A_3}}_{n-k-1} ,
 \end{equation}
where
\begin{equation}
C_{\vec{x}_i,x_j}=
      \dfrac{1}{\sqrt{\dloc}}\sum_{\ell \in [\dloc ]}
     \ket{\ell +x_j}
        \otimes 
\omega^{(\vec{x}_i A_1+ a x_j) \ell} \ket{\ell}.
\end{equation}
Clearly, the support of a state \cref{comp2} equals to $\dloc^{k+1}$. As a consequence, the support of states in \cref{star} and \cref{comp2} differs. Note that both equations present the same state $\ket{\phi_{n,2}}$, hence should have the same support. This ends the proof.
\end{proof}

In such a way, we have shown that two families of AME($n$,$\dloc$) states: $\ket{\phi_{n,0}}$ and $\ket{\phi_{n,2}}$ for $n=2k+1$ are not SLOCC-equivalent and hence  LU equivalent.

\end{widetext}

\end{document}